\documentclass [11pt] {article}
\usepackage[utf8]{inputenc}
\usepackage{csquotes}
\usepackage[draft]{hyperref}

\usepackage[style=apa,backend=biber]{biblatex}
\usepackage{amssymb} 
\usepackage{amsmath}
\usepackage{mathrsfs} 
\usepackage{blindtext}
\usepackage[english]{babel}
\usepackage{graphicx}
\usepackage{physics}
\usepackage[toc,page]{appendix} 
\usepackage{url}
\usepackage{mathtools}
\usepackage{setspace}
\usepackage{subfig}
\usepackage{pdfpages}
\usepackage{afterpage}
\usepackage{enumitem}
\usepackage{float}
\usepackage{enumerate}
\usepackage{kantlipsum}
\usepackage{multirow}
\usepackage{threeparttable}
\usepackage{placeins}
\usepackage{afterpage}

\newcommand{\comment}[1]{}

\newcommand\T{\rule{0pt}{2.3ex}}

\usepackage[margin=1.1in]{geometry}
\newlength\mystoreparindent

\usepackage{chngcntr}
\counterwithout{equation}{section}

\long\def\/*#1*/{}

\newtheorem{thm}{Theorem}

\newtheorem{lem}[thm]{Lemma}
\newtheorem{assum}[thm]{Assumption}

\newtheorem{corol}[thm]{Corollary} 
\newenvironment{proof}[1][Proof]{\begin{trivlist}
\item[\hskip \labelsep {\bfseries #1}]}{\end{trivlist}}

\numberwithin{thm}{section}
\numberwithin{prf}{section}
\numberwithin{equation}{section}
\counterwithin{table}{section}
\counterwithin{figure}{section}

\newcommand*{\QEDA}{\hfill\ensuremath{\blacksquare}}

\comment{
\usepackage{fancyhdr}
\pagestyle{fancy}
\fancyhf{}
\rhead{Dissertation}
\lhead{}
\rfoot{Page \thepage}
}

\makeatletter
\newcommand*\bigcdot{\mathpalette\bigcdot@{.5}}
\newcommand*\bigcdot@[2]{\mathbin{\vcenter{\hbox{\scalebox{#2}{$\m@th#1\bullet$}}}}}
\makeatother

\usepackage{hyperref}
\hypersetup{
    colorlinks=true,
    linkcolor=black,
    filecolor=black,      
    urlcolor=black,
}
\newcommand{\mcG}{\mathcal{G}}

\newcommand{\mcR}{\mathcal{R}}

\newcommand{\mcW}{\mathcal{W}}

\newcommand{\mbE}{\mathbb{E}}

\newcommand{\vp}{\varphi}

\newcommand{\sumin}{\sum_{i=1}^n}

\vspace{-50mm}
\title{\LARGE{Overidentification testing with weak instruments \\ and heteroskedasticity\thanks{We thank Vincent Han, Stefan Hubner, Gregory Jolivet, Pascal Lavergne, David Pacini, Pietro Spini, Senay Sokullu and participants in seminars at the University of Bristol and University College London for helpful comments and suggestions. Contact Stuart Lane via email for MATLAB replication code, and preliminary versions of Python and R functions and packages for robust overidentification testing. \\
\indent \, *This paper is based on Chapters 1 and 2 of Stuart Lane's Ph.D. thesis (University of Bristol, 2024).}}}
\vspace{5mm}
\author{\Large{Stuart Lane}\thanks{School of Economics, University of Bristol, UK, stuart.lane@bristol.ac.uk} \hspace{10mm} \Large{Frank Windmeijer}\thanks{Department of Statistics and Nuffield College, University of Oxford, UK, frank.windmeijer@stats.ox.ac.uk} \\ \\}
\date{}

\begin{document}

\maketitle

\begin{abstract}
Exogeneity is key for IV estimators, which can assessed via overidentification (OID) tests. We discuss the Kleibergen-Paap ($KP$) rank test as a heteroskedasticity-robust OID test and compare to the typical $J$-test. We derive the heteroskedastic weak-instrument limiting distributions for $J$ and $KP$ as special cases of the robust score test estimated via 2SLS and LIML respectively. Monte Carlo simulations show that $KP$ usually performs better than $J$, which is prone to severe size distortions. Test size depends on model parameters not consistently estimable with weak instruments, so a conservative approach is recommended. This generalises recommendations to use LIML-based OID tests under homoskedasticity. We then revisit the classic problem of estimating the elasticity of intertemporal substitution (EIS) in lifecycle consumption models. Lagged macroeconomic indicators should provide naturally valid but frequently weak instruments. The literature provides a wide range of estimates for this parameter, and $J$ frequently rejects the null of valid instruments. $J$ often rejects the null whereas $KP$ does not; we suggest that $J$ over-rejects, sometimes severely. We argue that $KP$-test should be used over the $J$-test. We also argue that instrument invalidity/misspecification is unlikely the cause of the range of EIS estimates in the literature. \\

\end{abstract}

\thispagestyle{empty}

\newpage

\section{Introduction}\label{section introduction}

\vspace{3mm}

Exogeneity is a key identifying assumption for validity in instrumental variables (IV) regression. Overidentification tests assess this assumption with the hypothesis specification $\mathbb{H}_0:\mathbb{E}[z_iu_i]=0$ v.s. $\mathbb{H}_1:\mathbb{E}[z_iu_i]\neq0$, where $z_i$ and $u_i$ are the instrument vector and structural error for observation $i$, respectively. These tests can be performed whenever $k_z>k_x$, where $k_z$ and $k_x$ are the number of instruments and endogenous regressors respectively. The Sargan test \parencite{sargan1958} is an early example of such a test, testing the orthogonality of instruments in the standard linear setting with homoskedasticity. The Hansen test \parencite{hansen1982}, denoted by $J$, generalises the Sargan test to the generalised method of moments (GMM) framework, allowing for overidentification testing on a far greater range of econometric models. The properties of these test statistics are well-known under standard asymptotics (e.g. \textcite{greene2003}).

In this paper we discuss the use of the Kleibergen-Paap ($KP$ hereafter) rank test \parencite{kleibergen2006} as a heteroskedasticity-robust overidentification test and compare its suitability to the ubiquitous $J$-test. This is the first assessment of the test in an overidentification context, following from \textcite{windmeijer2018} who presents a framework linking overidentification tests, underidentification tests, the score test, and rank tests. Underidentification tests assess rank hypotheses such as $\mathbb{H}_0:\textup{rank}(\Pi)=k_x-1$ v.s. $\mathbb{H}_1:\textup{rank}(\Pi)=k_x$, where $\Pi$ is the $k_z \times k_x$ first-stage parameter matrix (the $KP$-test is a standard underidentification test, commonly reported in statistical packages such as Stata). By formulating overidentification tests as rank tests on the reduced-form $k_z \times (k_x+1)$ parameter matrix $\bar{\Pi}$ with $\mathbb{H}_0:\textup{rank}(\bar{\Pi})=k_x$ v.s. $\mathbb{H}_1:\textup{rank}(\bar{\Pi})=k_x+1$ instead of the usual  $\mathbb{H}_0:\mathbb{E}[z_iu_i]=0$, \textcite{windmeijer2018} shows that any underidentification test can be used as an overidentification test if performed on an appropriate auxiliary regression, with an analogous result for using overidentification tests as underidentification tests. The paper further shows that the 2SLS-based and LIML-based robust score tests are equivalent to the $J$- and $KP$-tests respectively. The $J$- and $KP$-tests have the same limiting $\chi^2(k_z-k_x)$ distributions under the null when instruments are strong and provide numerically similar values in finite samples. However, their behaviour in a weak instrument setting is unlikely to be similar, and the relative performance of IV methods when instruments are potentially weak is often a key motivator for choosing between them.

\textcite{staiger1997} derive the limiting distribution of the Sargan test under weak instruments with homoskedasticity, providing a general recommendation to use LIML over 2SLS for overidentification tests based on numerical evidence. We extend these results by deriving the limiting distribution of the robust score test (which nests $J$ and $KP$) under heteroskedastic weak instruments. Due to the highly non-standard nature of the limiting distributions, a direct comparison is difficult. We therefore conduct an extensive Monte Carlo study, and usually find $KP$ preferable to $J$. The $J$-test performs extremely poorly in models with high endogeneity and/or strong heteroskedasticity and in general $KP$ exhibits better size properties and is much less likely than $J$ to be severely size distorted, especially when the number of overidentifying restrictions is small. Test size is dependent on model parameters that are not consistently estimable with weak instruments such as the strength of endogeneity. Consequently, it is difficult for researchers to gauge which test statistic would work better based on parameters estimated from data, and recommend a conservative approach that avoids severe size distortions where possible. We therefore find that our guidance to use the $KP$-test (equivalent to a LIML robust score test) under heteroskedastic weak instruments generalises the guidance from \textcite{staiger1997} for LIML-based testing under homoskedasticity. 

Given these theoretical results, we revisit the classic macroeconomic problem of estimating the elasticity of intertemporal substitution (EIS), the degree to which consumers adjust their consumption path in response to changes in the expected real interest rate. Due to the importance of this parameter for both theory and policy, it is unsurprising that the EIS has been given considerable attention in the empirical literature (e.g. \textcite{hansen1983} and \textcite{hall1988}). However, there is large variation in values of the EIS that constitute credible estimates, with substantial variation depending on the country, asset of choice or how the data are aggregated. Common concerns regard the strength and validity of the instruments used in estimation, often formed from lagged macroeconomic indicators. Such instruments should be naturally exogenous, but unfortunately are often poorly correlated with endogenous variables at time $t$, particularly as the number of lags increases, and this can lead to substantial weak instrument problems \parencite{yogo2004}. The model also likely suffers from heteroskedasticity, which represents precautionary saving in this model \parencite{yogo2004, gomes2013}. 

Within the literature, there are a number of papers where both weak instruments and heteroskedasticity are present, and that reject the overidentifying restrictions using the $J$-test e.g. \textcite{epstein1991}, \textcite{dacy2011}, \textcite{gomes2011}, and \textcite{gomes2013}. This leads to doubt over the validity of the moment conditions and instruments that identify and estimate the model and could potentially, at least partially, explain the variation in estimates of the EIS found in the literature. However, our simulation results suggest that $J$ performs poorly under heteroskedastic weak instruments, with often large over-rejections of the null hypothesis, whereas the $KP$-test does not, which may cast doubt over the rejections of the overidentifying restrictions seen in the literature.

In light of our theoretical and numerical evidence, we compare the performance of the $J$-test and the $KP$-test using the datasets of \textcite{yogo2004} and \textcite{pozzi2022}. These two applications provide different insights; the \textcite{yogo2004} application allows us to clearly see the impact of weak instruments via estimating two specifications, by switching between consumption and the real interest rate to be the independent variable. The instrument set with the second normalisation is substantially weaker. On the other hand, the \textcite{pozzi2022} application allows us to assess test performance with two different sets of instruments, where one set of instruments is more plausibly valid but weaker, and the other is less plausibly valid but stronger. 

We provide empirical evidence that the $J$-test over-rejects the overidentifying restrictions, whereas the $KP$-test rejects at a frequency consistent with a nominal 5\% level. We find similar patterns of results using datasets from the two papers and argue that the $KP$-test is more reliable than the $J$-test. Both tests provide very similar numerical values when instruments are strong and both do not reject the overidentifying restrictions. However, when instruments are possibly weak, $KP$ continues to not reject the overidentifying restrictions, whereas the $J$-test rejects frequently. We suggest that the substantial number of rejections seen in the literature can be accounted for by the poor behaviour of the $J$-test, rather than there being issues with the instruments or model specification, with the key implication of this being that instrument invalidity is unlikely the cause of variation in EIS obtained in the literature.

In the following section, we present the linear model and assumptions. We also introduce the relevant estimators and test statistics, and describe the framework linking overidentification, underidentification and rank tests. Section \ref{section weakhet} focuses on a theoretical analysis of the estimators and test statistics under heteroskedastic weak instruments. Section \ref{section mc} assesses the behaviour of the test statistics in Monte Carlo simulations, and provides numerical evidence that the $KP$-test typically performs better than the $J$-test under heteroskedastic weak instruments. In Section \ref{sec emp model} we provide an overview of the basic lifecycle consumption model, and the moment restrictions to be tested. In Section \ref{sec yogo}, we compare the performance of the $J$- and $KP$-tests using the dataset of \textcite{yogo2004} and suggest that $KP$ has the superior performance. We defer the application using the \textcite{pozzi2022} dataset to the Appendix purely for space considerations, but emphasise that this application is as valuable and insightful as the \textcite{yogo2004} application. All proofs and additional simulations are likewise deferred to the Appendix.

\section{Linear IV model, estimators and test statistics}\label{section model}

Consider the linear IV model
\begin{align}
    y &= X\beta + u, \label{eq model structural} \\
    X &= Z\Pi + V, \label{eq model firststage}
\end{align}
where $y$ is the $n\times 1$ dependent variable, $X$ is the $n\times k_x$ endogenous regressor matrix, $Z$ is the $n\times k_z$ instrument matrix (with $i^{th}$ row $z_i'$), $u$ is the $n\times 1$ structural error term and $V$ is the $n\times k_x$ vector of first-stage errors.
$\beta$ is the structural parameter of interest, $\Pi$ is the $k_z\times k_x$ matrix of first-stage parameters, and define $W=[ y \ \  X ]$ (with $i^{th}$ row $w_i'$). Additional exogenous regressors $\tilde{X}$ can be included (\ref{eq model structural}) and (\ref{eq model firststage}) and then easily partialled out. The reduced form equation is obtained by substituting (\ref{eq model firststage}) into (\ref{eq model structural}) to yield
\begin{equation}\label{eq model reducedform}
    W = Z \bar{\Pi} + \bar{V} 
\end{equation}
where $\bar{\Pi} = [ \Pi_y \ \ \Pi ]$ and $\bar{V} = [ V_y \ \ V ]$, with $\Pi_y=\Pi\beta$ and $V_y=u+V\beta$. Equation (\ref{eq model reducedform}) collects all endogenous variables on the left-hand side and regresses them on the instruments. 
We make the following assumptions:

\begin{assum}\label{as weakhom xziid}
$k_z>k_x$ is fixed.
\end{assum}

\begin{assum}\label{as weakhet qzz}
$Z'Z/n \overset{p}{\to} \mathbb{E}[z_iz_i']=Q_{ZZ}$. $Q_{ZZ}$ has full column rank $k_z$.
\end{assum}

\begin{assum} \label{as weakhet errors}
$\mathbb{E}[u_i^2]= \tilde{\sigma}_u^2$, $\mathbb{E}[v_iv_i'] = \tilde{\Sigma}_V$ and $\mathbb{E}[v_iu_i] = \tilde{\Sigma}_{Vu}$. Further, $\mbE[z_iu_i] = 0$, $\mathbb{E}[u_i^2|z_i]$, $\mathbb{E}[v_iv_i'|z_i]$ and $\mathbb{E}[v_iu_i|z_i]$ are finite for each $i$. Also
\begin{align*}
\begin{pmatrix}
    \frac{1}{\sqrt{n}} Z'u \\
    \frac{1}{\sqrt{n}} Z'V
    \end{pmatrix}
   &\overset{d}{\to}
   \begin{pmatrix}
    \Psi_{Zu}^* \\
    \Psi_{ZV}^*
    \end{pmatrix},
\end{align*}
where
\begin{align*}
   \begin{pmatrix}
       \Psi_{Zu}^* \\
       \textup{vec}(\Psi_{ZV}^*)
   \end{pmatrix} \sim N(0,\Omega_Z), \, \, \, \,
   \Omega_{Z} = 
   \begin{pmatrix}
       \Omega_{Zu} & \Omega_{Z,Vu}' \\
       \Omega_{Z,Vu} & \Omega_{ZV}
   \end{pmatrix} = 
   \begin{pmatrix}
      \mathbb{E}[u_i^2 z_i z_i'] & \mathbb{E}[u_iv_i' \otimes z_i z_i'] \\
      \mathbb{E}[v_iu_i \otimes z_i z_i'] & \mathbb{E}[v_iv_i' \otimes z_i z_i']
   \end{pmatrix}. 
\end{align*}
where $\textup{vec}(\cdot)$ is the vectorisation operator and $\otimes$ is the Kronecker product.
\end{assum}
Assumption \ref{as weakhom xziid} assumes overidentification to allow for specification testing, with the number of instruments fixed; this rules out many-instrument asymptotics. Assumption \ref{as weakhet qzz} ensures that $Z'Z/n$ converges to a constant matrix with full rank. Assumption \ref{as weakhet errors} assumes the instrument is valid (and this is the assumption we wish to test). Further, it is general and no specific relationship between the second moments of the errors and instruments is assumed, except that they are finite. Because of the conditional heteroskedasticity, $\Omega_Z$ will in general lack the Kronecker structure seen under homoskedasticity. The assumption states that $\Psi_{Zu}^*$ is a multivariate normal variable with distribution $N(0,\mbE[u_i^2 z_i z_i'])$, and $\Psi_{ZV}^*$ is a matrix normal variable such that $\textup{vec}(\Psi_{ZV}^*)\sim N(0,\mbE[v_iv_i'\otimes z_iz_i'])$. These assumptions are standard.

\subsection{Estimators}\label{section model estimators}

A wide range of estimators are available for IV models. For simplicity, we focus on 2SLS and LIML (denoted by $\hat{\beta}_{2SLS}$ and $\hat{\beta}_{L}$ respectively), defined as
\begin{equation}\label{eq 2slsformula}
    \hat{\beta}_{2SLS} = (X'P_ZX)^{-1}X'P_Zy \  \textup{  and  }
\end{equation}
\begin{equation}\label{eq limlformula}
    \hat{\beta}_{L} = (X'P_ZX - \hat{\alpha}_LX'X)^{-1}(X'P_Zy - \hat{\alpha}_LX'y),
 \end{equation}
where $P_Z = Z(Z'Z)^{-1}Z'$ and $\hat{\alpha}_L$ is the smallest root of the characteristic polynomial $|W'P_ZW - \alpha W'W| = 0$. 2SLS is the most common IV estimator seen in the applied literature and has been extensively studied within the econometrics literature. The LIML estimator solves the maximum likelihood problem for a single equation within a system of endogenous simultaneous equations \parencite{anderson1949}. LIML has also been studied extensively within the theoretical literature, but has not seen the widespread use that 2SLS has enjoyed in the applied literature, despite exhibiting favourable finite-sample properties. LIML often outperforms 2SLS in finite samples with strong instruments and performs better in homoskedastic weak-instrument settings. Both estimators are part of the wider class of estimators of the form
\begin{equation}\label{eq weakhom kclass}
        \hat{\beta}(\alpha) = (X'P_ZX - \alpha X'X)^{-1}(X'P_Zy - \alpha X'y),
\end{equation}
where clearly $\alpha=0$ for 2SLS and $\alpha = \hat{\alpha}_L$ for LIML.\footnote{Since 2SLS and LIML relate to the $J$ and $KP$ tests respectively, it is simple to study these two estimators and subsequently study test statistics that applied researchers are already familiar with and for which implementation in statistical software already exists. However, generalising our results to other estimators could be an interesting topic for future research.}

\subsection{Test statistics}\label{section model tests}

All standard IV estimators require instrument exogeneity for consistency. While this assumption is impossible to test directly, when instruments outnumber the endogenous regressors, we can test for the validity of the overidentifying restrictions. Through this, we may be able to find evidence that the restrictions do not hold, and it is recommended to conduct such a test whenever possible. The null and alternative hypotheses are usually specified as
\begin{equation}
    \mathbb{H}_0 : \mathbb{E}[z_iu_i]= 0 \textup{ v.s. } \mathbb{H}_1 : \mathbb{E}[z_iu_i]\neq 0,
\end{equation}
so rejection provides evidence that exogeneity is not satisfied and/or the model is misspecified. For this, it is natural to consider two-step estimators and test statistics when heteroskedasticity may be present. For the general linear GMM estimator $\hat{\beta}_{GMM}=(X'Z\mcG^{-1} Z'X)^{-1}X'Z\mcG^{-1} Z'y$, where $\mcG$ is some positive-definite weighting matrix, asymptotic efficiency is achieved if $\mcG\overset{p}{\to}\Omega_{Zu}$. Consistency is easily achieved with typical plug-in estimators. The two-step GMM estimator is given by
\begin{equation}\label{eq esttest 2step}
    \hat{\beta}_2 = (\hat{\Pi}_1'Z'Z(Z'H_{\hat{u}_1}Z)^{-1}Z'X)^{-1}\hat{\Pi}_1'Z'Z(Z'H_{\hat{u}_1}Z)^{-1}Z'y
\end{equation}
where $\hat{\Pi}_1$ is the appropriate first-stage estimator for $\hat{\beta}_1$ (e.g. if $\hat{\beta}_1 = \hat{\beta}_{2SLS}$, then $\hat{\Pi}_{2SLS}=(Z'Z)^{-1}Z'X$ and if $\hat{\beta}_1 = \hat{\beta}_{L}$, then $\hat{\Pi}_{L}=(Z'M_{\hat{u}_L}Z)^{-1}Z'M_{\hat{u}_L}X$, with $M_A = I_n - P_A$ and $P_A = A(A'A)^{-1}A'$ for some generic $n \times L$ matrix $A$) and $Z'H_{\hat{u}_1}Z/n\overset{p}{\to}\Omega_{Zu}$. Under standard conditions, both $\hat{\beta}_{2,2SLS}$ and $\hat{\beta}_{2,L}$ are asymptotically efficient and equivalent to the infeasible estimator $\hat{\beta}_{opt}$, given by $\sqrt{n}(\hat{\beta}_{opt}-\beta)\overset{d}{\to}N(0,\Pi'Q_{ZZ}\Omega_{Zu}^{-1}Q_{ZZ}\Pi)$, computable with complete information on $\Pi$ and $\Omega_{Zu}$ \parencite{windmeijer2018}. 

The most common heteroskedasticity-robust two-step test statistic is the $J$-test \parencite{hansen1982}, given as
\begin{equation}\label{eq esttest hansen}
    J(\hat{\beta}_1,\hat{\beta}_2)= \hat{u}_2'Z(Z'H_{\hat{u}_1}Z)^{-1}Z'\hat{u}_2,
\end{equation}
where $\hat{u}_2=y-X\hat{\beta}_2$. Standard implementation sets $\hat{\beta}_1 = \hat{\beta}_{2SLS}$, and under the null $J(\hat{\beta}_1,\hat{\beta}_2)\overset{d}{\to}\chi^2(k_z-k_x)$ with general heteroskedasticity and/or autocorrelation. Two-step estimation can also be used to provide a heteroskedasticity-robust score test. Partition the instrument matrix into $Z = [ Z_1 \ \ Z_2 ]$, where $Z_1$ is an $n\times k_x$ matrix of just-identifying instruments and $Z_2$ is an $n\times(k_z-k_x)$ matrix of remaining overidentifying instruments. Then the robust score statistic is given by
\begin{equation}\label{eq esttest twostepscore}
    S_r(\hat{\beta}_1,\hat{\beta}_2) = \hat{u}_2'M_{\hat{X}}Z_2(Z_2'M_{\hat{X}}H_{\hat{u}_1}M_{\hat{X}}Z_2)^{-1}Z_2'M_{\hat{X}}\hat{u}_2
\end{equation}
where $\hat{X}=Z\hat{\Pi}$, with $\hat{\Pi}$ the appropriate estimator for $\Pi$ depending on choice of $\hat{\beta}_1$. Again, $S_r(\hat{\beta}_1,\hat{\beta}_2)\overset{d}{\to}\chi^2(k_z-k_x)$ under usual asymptotics. \textcite{windmeijer2018} proves that (\ref{eq esttest hansen}) and (\ref{eq esttest twostepscore}) are equivalent. However, the paper also shows that the two-step approach is actually redundant, since $Z_2'M_{\hat{X}}\hat{u}_j = Z_2'M_{\hat{X}}(y - x\hat{\beta}_j) = Z_2'M_{\hat{X}}y$ for $j\in\{1,2\}$. It therefore suffices to use a one-step robust score test given by
\begin{equation}\label{eq esttest robscore}
    S_r(\hat{\beta}_1) = \hat{u}_1'M_{\hat{X}}Z_2(Z_2'M_{\hat{X}}H_{\hat{u}_1}M_{\hat{X}}Z_2)^{-1}Z_2'M_{\hat{X}}\hat{u}_1.
\end{equation}
with $S_r(\hat{\beta}_1,\hat{\beta}_2) = S_r(\hat{\beta}_1)$ and consequently $J(\hat{\beta}_1,\hat{\beta}_2)=S_r(\hat{\beta}_1)$. Since the tests are numerically equivalent, we will ignore two-step estimators and testing and consider only $\hat{\beta}_1$ and $S_r(\hat{\beta}_1)$ for simplicity. 

\subsection{Relation to rank tests}\label{subsection ranktests}

To demontstrate the link between overidentification and rank tests, it is useful to first consider general underidentification tests.
The null and alternative hypotheses for underidentification tests are usually specified as
\begin{equation}\label{eq esttest uidnull}
    \mathbb{H}_0:\textup{rank}(\Pi)=k_x-1 \textup{ v.s. } \mathbb{H}_1:\textup{rank}(\Pi)=k_x.
\end{equation}
Identification of $\beta$ requires that $\mathbb{E}[Z'X]$ has full column rank, such that $\Pi=\mathbb{E}[Z'Z]^{-1}\mathbb{E}[Z'X]$ has full column rank; $\Pi$ is reduced-rank under the null. Therefore, there exists some non-zero vector $\zeta\in\mathbb{R}^{k_x}$ such that $\mathbb{E}[Z'X]\zeta=0$.\footnote{$A\in\mathbb{R}^{n\times m}$ is rank deficient iff there exists some non-zero vector $\zeta\in \mathbb{R}^{m}$ that satisfies $A\zeta=0$. The equivalence of underidentification testing and the standard $F$-test for $\mathbb{H}_0:\Pi = 0$ in the $k_x=1$ case is clear.}

To link this with the standard overidentification null of $\mathbb{H}_0:\mathbb{E}[z_iu_i]=0$ v.s. $\mathbb{H}_1:\mathbb{E}[z_iu_i]\neq 0$, consider $u_i = y_i - x_i'\beta$ = $w_i'\psi$, where $\psi=( 1 \ \ -\beta' )'$. It follows that $\mathbb{E}[z_iu_i]=0$ is equivalent to $\mathbb{E}[z_iw_i'\psi]=0$. Since $\mathbb{E}[z_iw_i']$ is an $n\times(k_x+1)$ matrix and $\psi$ is a non-zero $(k_x+1)$-vector with first element normalised to 1, under the null hypothesis of $\mathbb{E}[z_iw_i'\psi]=0$, $\mathbb{E}[z_iw_i']$ must be reduced-rank as this is the only way the homogeneous system of equations can be solved. This in turn implies that the null hypothesis $\mathbb{E}[z_iu_i]=0$ is satisfied if and only if $\mathbb{E}[z_iw_i']$ is reduced-rank. Just as the rank of $\mathbb{E}[z_ix_i']$ determines the rank of $\Pi$ given Assumptions \ref{as weakhom xziid} and \ref{as weakhet qzz}, the rank of $\mathbb{E}[z_iw_i']$ determines the rank of the reduced-form parameter matrix $\bar{\Pi}$. Therefore, the familiar overidentification null of the orthogonality of the instruments and structural errors can be restated as a test of rank on the reduced-form parameter matrix as 
\begin{equation}\label{eq esttest oidnull}
\mathbb{H}_0:\textup{rank}(\bar{\Pi})=k_x \textup{ v.s. } \mathbb{H}_1:\textup{rank}(\bar{\Pi})=k_x+1.
\end{equation} 
It follows from this result that rank tests such as the $KP$-test can be used to assess instrument validity.

To link rank testing with the robust score test in (\ref{eq esttest robscore}), create the partition $Z=[ Z_1 \ \ Z_2 ]$, where $Z_1$ is an $n\times k_x$ matrix of just-identifying instruments and $Z_2$ is an $n\times(k_z-k_x)$ matrix of overidentifying instruments. With conformable partitions $\Pi_y = [\Pi_{y,1}' \ \ \Pi_{y,2}' ]'$ and $\Pi = [\Pi_1' \ \ \Pi_2' ]'$, the reduced-form equation $y=Z\Pi_y+V_y$ can be expressed as
\begin{align}\label{eq scoregammaeq}
    y &= Z_1 \Pi_{y,1} + Z_2 \Pi_{y,2} + V_y \notag \\
    &= Z_1 \Pi_1 \Pi_1^{-1}\Pi_{y,1} + Z_2 \Pi_{y,2} + V_y \notag \\
    &= (X - Z_2\Pi_2 - V)\Pi_1^{-1}\Pi_{y,1} + Z_2 \Pi_{y,2} + V_y \notag \\
    &= X\tilde{\beta} + Z_2\theta + \eta,
\end{align}
where $\tilde{\beta}=\Pi_1^{-1}\Pi_{y,1}$, $\theta = \Pi_{y,2}-\Pi_2\Pi_1^{-1}\Pi_{y,1}$, with $\Pi_1$ nonsingular, and $\eta = V_y - V\Pi_1^{-1}\Pi_{y,1}$. The first line expands $Z\Pi_y = [ Z_1 \ \ Z_2 ][\Pi_{y,1}' \ \ \Pi_{y,2}' ]'$ to $Z_1 \Pi_{y,1} + Z_2 \Pi_{y,2}$, the second line multiplies by $\Pi_1$ and $\Pi_1^{-1}$ and the third line substitutes in the expression for $Z_1\Pi_1$. Under correct model specification, $\Pi_y=\Pi\beta$ follows from (\ref{eq model reducedform}), which implies that $[\Pi_{y,1}' \ \ \Pi_{y,2}' ]' = [\Pi_1\beta \ \ \Pi_2\beta ]'$. Therefore, we can substitute $\Pi_{y,1}=\Pi_1\beta$ into $\tilde{\beta}$ to yield $\tilde{\beta} = \beta$ and similarly $\theta = \Pi_2\beta - \Pi_2\Pi_1^{-1}(\Pi_1\beta)= \Pi_2\beta -\Pi_2\beta=0$. So, testing the hypothesis $\mathbb{H}_0:\theta=0$ is equivalent to testing for correct specification. This null can be tested using the robust score test in (\ref{eq esttest robscore}). In particular,  \textcite{windmeijer2018} shows that the $KP$-test for $\mathbb{H}_0:\textup{rank}(\bar{\Pi})=k_x$ v.s. $\mathbb{H}_0:\textup{rank}(\bar{\Pi})=k_x+1$ is equivalent to the robust score test estimated via LIML for $\mathbb{H}_0:\theta=0$ v.s.  $\mathbb{H}_1:\theta\neq 0$ in (\ref{eq scoregammaeq}), given by
\begin{equation}
     S_r(\hat{\beta}_L) = \hat{u}_L'M_{\hat{X}_L}Z_2\left(Z_2'M_{\hat{X}_L}H_{\hat{u}_L}M_{\hat{X}_L}Z_2\right)^{-1}Z_2'M_{\hat{X}_L}\hat{u}_L,  
\end{equation}
where $\hat{X}_L = Z\hat{\Pi}_L = Z(Z'M_{\hat{u}_L}Z)^{-1}Z'M_{\hat{u}_L}X$.

Under strong instruments, $J$ and $KP$ perform similarly. Figure \ref{graph esttest power} shows the finite-sample power of the $J$- and $KP$-tests against local-to-zero alternatives across different strengths of endogeneity for a heteroskedastic strong-instrument model with a single overidentifying restriction. Increasing  $\alpha$ increases the strength of the heteroskedasticity. Both tests clearly perform similarly, but the $KP$-test has higher power in most of the set-ups shown; the only seeming exception is the high endogeneity design for $\omega>0$, but the $J$-test does not have correct size in this setup. The improvements in power of $KP$ over $J$ are larger in the design with the stronger heteroskedasticity. This demonstrates clear potential for the $KP$-test to be a strong alternative to the usual $J$-test typically reported in applied research when instruments are strong.

\begin{figure}[h!]
    \centering
    \subfloat[$\alpha = 0.5$] {{\includegraphics[width=8.2cm]{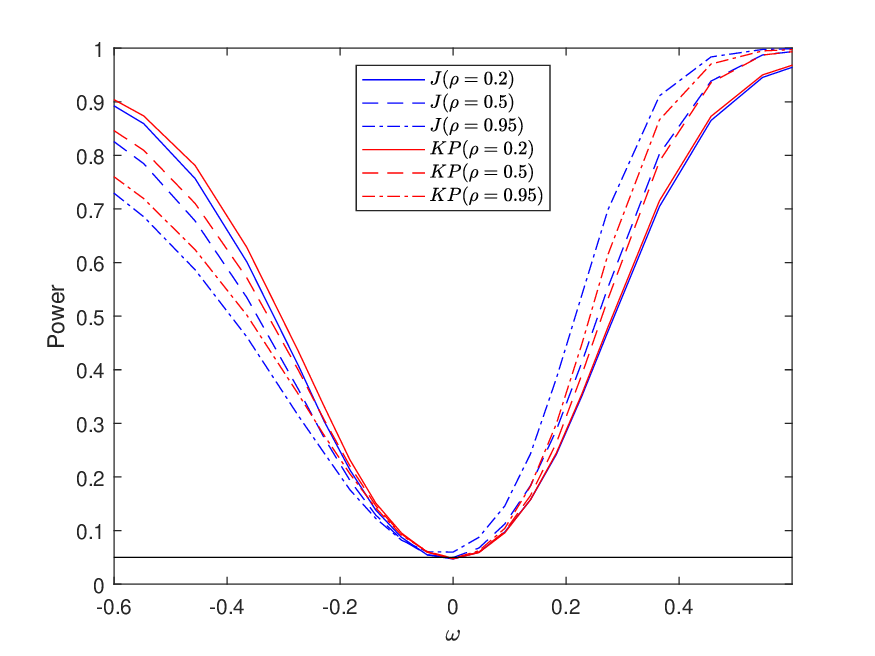}}}
    \subfloat[$\alpha = 1$] {{\includegraphics[width=8.2cm]{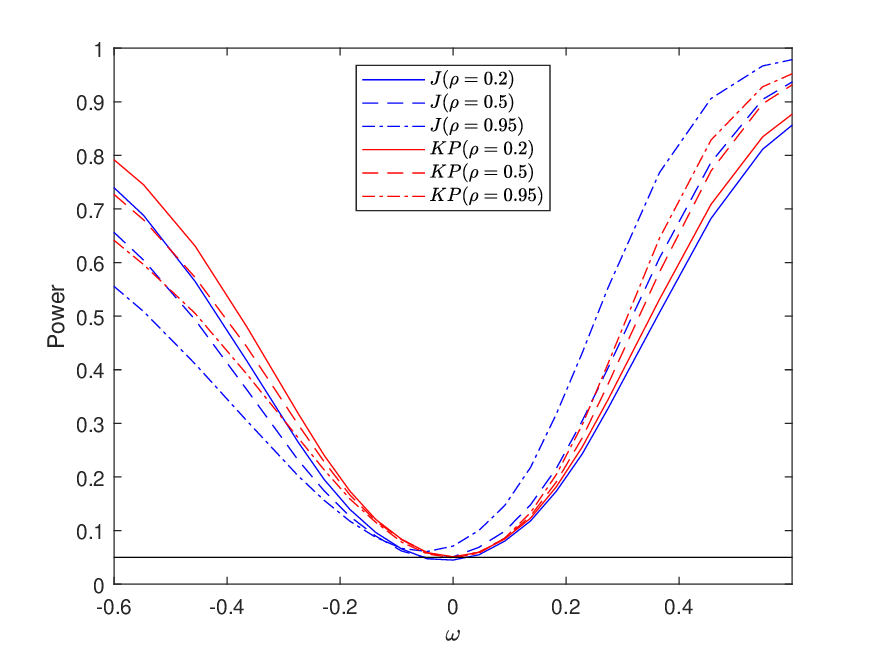}}}
    \caption[Power of the $J$- and $KP$-tests under strong instruments]{Power of $J$- and $KP$-tests. To generate the heteroskedastic errors, $(u^*_i,v^*_i)$ are generated jointly standard normal with correlation $\rho$. Then $u_i = z_{1,i}\omega + |z_{1,i}|^{\alpha}u^*_i$ and $v_i = |z_{1,i}|^{\alpha}v^*_i$ for $\alpha\in\{0.5,1\}$ (the absolute values to ensure no complex-valued errors). The overidentification test null is equivalent to $\mathbb{H}_0 : \omega = 0$. We set $n=120$, $\mu^2=48$ and use 20,000 repetitions. See Section \ref{section mc} for full details of set-up.}
    \label{graph esttest power}
\end{figure}

However, the more interesting discussion is how $J$- and $KP$ behave under heteroskedastic weak instruments, as performance when identification strength is questionable is often an important motivation behind selecting methods for estimation and inference with instrumental variables. LIML often has favourable properties relative to 2SLS under homoskedastic weak instruments, so we aim to see how this translates into the heteroskedastic case. Further, \textcite{staiger1997} recommend in general to use a LIML-based Sargan test for overidentification testing, due to the sensitivity of 2SLS in models with high endogeneity or a large number of overidentifying restrictions.

\section{Weak instruments under heteroskedasticity}\label{section weakhet}
To consider estimation and testing with heteroskedastic weak instruments, we require the following assumption:

\begin{assum}\label{as weakhet pi}
$\Pi=C/\sqrt{n}$ for some finite ${k_z\times k_x}$ matrix $C$, with $\textup{rank}(C) = k_x$. 
\end{assum}
This assumption states that the first-stage parameter matrix is local-to-zero at rate $\sqrt{n}$, following \textcite{staiger1997}. We define the concentration matrix as
\begin{equation}\label{eq concpar}
    \mu^2 = k_z\mathcal{V}_{ZV}^{-1/2} \Pi'Z'Z\Pi \mathcal{V}_{ZV}^{-1/2} \overset{p}{\to} k_z\mathcal{V}_{ZV}^{-1/2} C'Q_{ZZ}C\mathcal{V}_{ZV}^{-1/2},
\end{equation}
where $\mathcal{V}_{ZV} = (I_{k_x} \otimes \textup{vec}(I_{k_z}))' (\Omega_{ZV} \otimes Q_{ZZ}^{-1}) (I_{k_x} \otimes \textup{vec}(I_{k_z}))$.\footnote{ When the instruments are orthonormalised, the minimum eigenvalue of $\mu^2/k_z$ gives the weak-instrument test statistic of \textcite{lewis2022}, which with $k_x = 1$ collapses to the test statistic of \textcite{montiel2013}.} It is simple to show that under homoskedasticity, where $\Omega_{ZV} = \Sigma_{V} \otimes Q_{ZZ}$ for $\mathbb{E}[v_iv_i'|z_i] = \Sigma_V$, then $\mathcal{V}_{ZV} = k_z\Sigma_V$, and (\ref{eq concpar}) collapses down to the usual homoskedastic concentration matrix $\mu^2 = \Sigma_V^{-1/2} \Pi'Z'Z\Pi \Sigma_V^{-1/2}\overset{p}{\to} \Sigma_V^{-1/2} C'Q_{ZZ}C\Sigma_V^{-1/2}$.

For heteroskedastic limiting distributions, the following lemma is required.

\begin{lem} \label{lem weakhet staigerstock}
Let Assumptions \ref{as weakhom xziid}, \ref{as weakhet qzz},  \ref{as weakhet errors} and \ref{as weakhet pi} hold. Then the following results hold: \\
(a) $Z'X/\sqrt{n} \overset{d}{\to} Q_{ZZ}C + \Psi^*_{ZV}$.  \\
(b) $X'P_ZX \overset{d}{\to} (Q_{ZZ}C + \Psi^*_{ZV})'Q_{ZZ}^{-1}(Q_{ZZ}C + \Psi^*_{ZV})$. \\
(c) $X'P_Zu \overset{d}{\to} (Q_{ZZ}C + \Psi^*_{ZV})'Q_{ZZ}^{-1}\Psi^*_{Zu}$. \\
(d) $X'X/n \overset{p}{\to} \mathbb{E}[v_iv_i'] = \tilde{\Sigma}_{V}$. \\
(e) $X'u/n \overset{p}{\to} \mathbb{E}[v_iu_i] = \tilde{\sigma}_{Vu}$.
\end{lem}
The interpretations of these results are clear e.g. (\textit{a}) implies that $Z'X/\sqrt{n}$ converges to a normal distribution, with asymptotic variance accounting for the structure of the heteroskedasticity. Parts (\textit{b}) and (\textit{c}) show that $X'P_ZX$ converges to a general noncentral correlated Wishart matrix and $X'P_Zu$ converges to a product of two matrix normals. Assume $Z'\bar{V}/\sqrt{n}\overset{d}{\to}\Psi_{Z\bar{V}}^*$ and let $W'W/n\overset{p}{\to}\Sigma_{\bar{V}}^*$, where $\Sigma_{\bar{V}}^*$ is the reduced-form variance matrix. Throughout the main body of the paper, we leave our limiting distributions in the random matrix form of Lemma \ref{lem weakhet staigerstock} for notational simplicity and clarity, although in Appendix \ref{app vectorised}, we restate the results of the paper in vectorised form.

\begin{lem}\label{lem weakhet alphal}
Let Assumptions \ref{as weakhom xziid}, \ref{as weakhet qzz},  \ref{as weakhet errors} and \ref{as weakhet pi} hold. Then for $B = [\beta \ \ I_{k_x}]$
\begin{equation}\label{eq weakhet alphaliml}
    n\hat{\alpha}_L\overset{d}{\to}\min_{||\phi||=1} \frac{\phi'[Q_{ZZ}CB + \Psi_{Z\bar{V}}^*]'Q_{ZZ}^{-1}[Q_{ZZ}CB + \Psi_{Z\bar{V}}^*]\phi}{\phi'\Sigma_{\bar{V}}^*\phi}=\tilde{\alpha}_L^*.
\end{equation}
\end{lem}
Expression (\ref{eq weakhet alphaliml}) shows that $\tilde{\alpha}^*_L$ is asymptotically random, unlike under standard asymptotics where $n\hat{\alpha}_L \overset{p}{\to} 0$. $\tilde{\alpha}^*_L$ is the minimum eigenvalue of a general noncentral correlated Wishart matrix. By definition, $\tilde{\alpha}_L^*$ is defined as the limit of
\begin{equation*}
    n\hat{\alpha}_L = \min_{||\phi||=1}\frac{\phi'W'P_ZW\phi}{\phi'\left(\frac{1}{n}W'W\right)\phi},
\end{equation*}
i.e. the smallest eigenvalue of $(\frac{1}{n}W'W)^{-1/2}(W'P_ZW)(\frac{1}{n}W'W)^{-1/2}$. Then
\begin{align}\label{eq mc wpwlimit1}
    W'P_ZW = \begin{bmatrix}
        y'P_Zy & y'P_ZX \\
        X'P_Zy & X'P_ZX
    \end{bmatrix}
    \overset{d}{\to}
    \begin{bmatrix}
        (\xi_1\beta + \xi_2)'(\xi_1\beta + \xi_2) & (\xi_1\beta + \xi_2)'\xi_1 \\
        \xi_1'(\xi_1\beta + \xi_2) & \xi_1'\xi_1 
    \end{bmatrix},
\end{align}
where $\xi_1 = Q_{ZZ}^{-1/2}(Q_{ZZ}C + \Psi^*_{ZV})$ and $\xi_2 = Q_{ZZ}^{-1/2}\Psi^*_{Zu}$.
By Assumption \ref{as weakhet errors}, the $k_z \times (k_x+1)$ matrix $\Xi = [(\xi_1\beta + \xi_2) \ \ \xi_1 ]$ is distributed $\Xi \sim N(\mathcal{M}_{\xi}, \mathcal{V}_{\xi})$ for $\mathcal{M}_{\xi} = Q_{ZZ}^{1/2}CB$ and $\mathcal{V}_{\xi}$ dependent on the specific structure of the heteroskedasticity. It therefore follows that
\begin{equation}
    \Xi'\Xi \sim W_{k_x+1}(k_z,\mathcal{V}_{\xi},\Lambda),
\end{equation}
where $W_q(p,V,\Theta)$ denotes the non-central Wishart distribution with dimension $q$, $p$ degrees of freedom, scaling matrix $V$ and non-centrality matrix $\Theta$, and $\Lambda=\mathcal{M}_{\xi}'\mathcal{M}_{\xi}$. We can also see that 
\begin{equation}\label{eq mc wwlimits1}
    \frac{1}{n}W'W \overset{p}{\to} \Sigma^*_{\bar{V}}.
\end{equation}
Therefore, $\tilde{\alpha}^*_L$ is the distribution of the minimum eigenvalue of the limit of the random matrix $(\frac{1}{n}W'W)^{-1/2}(W'P_ZW)(\frac{1}{n}W'W)^{-1/2}$ with distribution
\begin{equation}\label{eq mc wishartlimit}
    \left(\frac{1}{n}W'W\right)^{-1/2}(W'P_ZW)\left(\frac{1}{n}W'W\right)^{-1/2} \overset{d}{\to} W_{k_x+1}(k_z,\Sigma^{*-1/2}_{\bar{V}}\mathcal{V}_{\xi}\Sigma^{*-1/2}_{\bar{V}},\Sigma^*_{\bar{V}}\Lambda\Sigma^*_{\bar{V}}).
\end{equation}
Clearly, the structure of the heteroskedasticity plays an important role in this distribution, although general statements are more difficult here than in \textcite{staiger1997} due to the loss of the Kronecker variance structure. Further, there are no results for the distribution of extreme eigenvalues of general non-central, correlated Wishart matrices to the best of our knowledge, and deriving results on the distribution of extreme eigenvalues of such matrices is beyond the scope of this paper. This makes precise statements about the behaviour of LIML impossible. We can however see from inspection that $\tilde{\alpha}^*_L$ will be dependent on instrument strength, instrument numerosity and the endogeneity and heteroskedasticity of the model.

\begin{lem}\label{lem weakhet betalimits}
Let Assumptions \ref{as weakhom xziid}, \ref{as weakhet qzz},  \ref{as weakhet errors} and \ref{as weakhet pi} hold and assume $n\alpha\overset{d}{\to}\tilde{\alpha}^*$. Then, 
\begin{equation}\label{eq weakhet b2sls}
    \hat{\beta}_{2SLS} - \beta \overset{d}{\to} \tilde{\beta}^*_{2SLS} = [(Q_{ZZ}C + \Psi^*_{ZV})'Q_{ZZ}^{-1}(Q_{ZZ}C + \Psi^*_{ZV})]^{-1}(Q_{ZZ}C + \Psi^*_{ZV})'Q_{ZZ}^{-1}\Psi^*_{Zu},
\end{equation}
\begin{equation}\label{eq weakhet bliml}
    \hat{\beta}_{L} - \beta \overset{d}{\to} \tilde{\beta}^*_{L} = [(Q_{ZZ}C + \Psi^*_{ZV})'Q_{ZZ}^{-1}(Q_{ZZ}C + \Psi^*_{ZV}) - \tilde{\alpha}^*_L \tilde{\Sigma}_V]^{-1}[(Q_{ZZ}C + \Psi^*_{ZV})'Q_{ZZ}^{-1}\Psi^*_{Zu} - \tilde{\alpha}^*_L\tilde{\Sigma}_{Vu}].
\end{equation}
\end{lem}
This lemma gives a minor adaptation of Lemma 1 from \textcite{montiel2013}, who give the limiting distributions of 2SLS and LIML in a heteroskedastic model with orthonormalised instruments i.e. $Z'Z/n = I_{k_z}$. Although 2SLS and LIML are invariant to this transformation, (\ref{eq weakhet b2sls}) and (\ref{eq weakhet bliml}) do not require the orthonomalising transformation and so the result can be applied to any data satisfying the assumptions specified in Lemma \ref{lem weakhet betalimits}. The expressions are complicated mixtures and ratios of normals, similar to the limiting distributions given by \textcite{staiger1997}. Both the spread and location of the 2SLS and LIML estimators are affected by the structure of the heteroskedasticity, which cannot be consistently estimated under weak-instrument asymptotics. Before moving to the limiting distribution of the $J$ and $KP$-tests, we need the limiting distributions of the first-stage parameter estimators, as these feature in the limits of the variance terms in the test statistics.

\begin{lem}\label{lem weakhet pilimit}
Let Assumptions \ref{as weakhom xziid}, \ref{as weakhet qzz},  \ref{as weakhet errors} and \ref{as weakhet pi} hold. Then: \\
(\textit{a}) $\sqrt{n}\hat{\Pi}_{2SLS}\overset{d}{\to} \tilde{\Pi}^*_{2SLS} = Q_{ZZ}^{-1}(Q_{ZZ}C + \Psi^*_{ZV})$. \\
(\textit{b}) $\sqrt{n}\hat{\Pi}_{L}\overset{d}{\to} \tilde{\Pi}^*_{L} = \tilde{\Pi}^*_{2SLS} - Q_{ZZ}^{-1}[(\Psi_{Zu}^* - (Q_{ZZ}C + \Psi_{ZV}^*)\tilde{\beta}^*_L)(\tilde{\Sigma}_{Vu} - \tilde{\beta}^*_L\tilde{\Sigma}_{V})']/[ \tilde{\sigma}_{u}^2 - 2\tilde{\beta}^*_L\tilde{\Sigma}_{Vu} + {\tilde{\beta}^{*'}_L}\tilde{\Sigma}_{V}\tilde{\beta}^*_L].$
\end{lem}

The 2SLS estimator for $\Pi$ is asymptotically matrix-normal and with location matrix $C$. Therefore, while the estimator is consistent, it is asymptotically biased. $\tilde{\Pi}^*_L$ is difficult to analyse due to the second term in the expression subtracted from $\tilde{\Pi}^*_{2SLS}$, given that this term is a complicated ratio of random variables with the nonstandard distribution $\tilde{\beta}^*_L$ entering in a nonlinear manner (note that $\hat{\beta}_L$ and $\hat{\Pi}_L$ are estimated simultaneously, whereas of course $\hat{\Pi}_{2SLS}$ is estimated in the first stage and then $\hat{\beta}_{2SLS}$ is estimated in the second stage).

\subsection{Overidentification testing}\label{section weakhet od}

Now we derive the limiting distribution of the robust score test under heteroskedastic weak instruments, which will nest the weak-instrument limits for the $J$- and $KP$-tests. For the limiting distribution of the robust score test, let $\tilde{Z}_2=M_{\hat{X}}Z_2$ (with $i^{th}$ row $\tilde{z}_{2i}'$).

\begin{assum}\label{as weakhet ztilde}
For some estimator $\hat{\beta}$, let the following limits exist:

\begin{equation*}
    \frac{1}{n}\sumin u_i^2\tilde{z}_{2,i}\tilde{z}_{2,i}' \overset{p}{\to} \Omega_{\tilde{Z}_{2},u}, \,\,\,
    \frac{1}{n}\sumin [u_iv_i'(\hat{\beta}-\beta)]\tilde{z}_{2,i}\tilde{z}_{2,i}' \overset{d}{\to} \Omega_{\tilde{Z}_{2},uV\tilde{\beta}},
\end{equation*}
\begin{equation*}
    \frac{1}{n}\sumin [(\hat{\beta}-\beta)'v_iv_i'(\hat{\beta}-\beta)]\tilde{z}_{2,i}\tilde{z}_{2,i}' \overset{d}{\to} \Omega_{\tilde{Z}_{2},\tilde{\beta} V\tilde{\beta}}.
\end{equation*}
Then for $\hat{u} = y - X\hat{\beta}$, define $\Omega_{\tilde{Z}_2\hat{u}} = \Omega_{\tilde{Z}_{2},u} - 2\Omega_{\tilde{Z}_{2},uV\tilde{\beta}} + \Omega_{\tilde{Z}_{2},\tilde{\beta} V\tilde{\beta}}.$
Also, let $Z'Z_2/n\overset{p}{\to}Q_2$, $Z_2'Z_2/n\overset{p}{\to}Q_{22}$, $Z_2'u/\sqrt{n}\overset{d}{\to}\Psi_{2,u}^*$ and $Z_2'V/\sqrt{n}\overset{d}{\to}\Psi_{2,V}^*$.
\end{assum}

The first part of Assumption \ref{as weakhet ztilde} gives general limits for terms that will appear in the decomposition of the denominator of the robust score test. The second part of Assumption \ref{as weakhet ztilde} defines two other convergence results that will be important for the limiting distributions in Theorem \ref{thm robscore}. 

All three terms $\Omega_{\tilde{Z}_{2},u}$, $\Omega_{\tilde{Z}_{2},uV\tilde{\beta}}$ and $\Omega_{\tilde{Z}_{2},\tilde{\beta} V\tilde{\beta}}$ are random matrices by Lemma \ref{lem weakhet betalimits} and will result in our test statistics becoming ratios of distributions asymptotically, analogous to the \textcite{staiger1997} results for the homoskedastic Sargan test. The matrix $\tilde{Z}_2'\tilde{Z}_2/n$ is asymptotically non-degenerate with weak instruments. This follows since
\begin{align}\label{eq weakhet z2mxz2a}
    \tilde{Z}_2'\tilde{Z}_2 = Z_2'M_{\hat{X}}Z_2 = Z_2'Z_2 - Z_2'P_{\hat{X}}Z_2 = Z_2'Z_2 - Z_2'Z\hat{\Pi}(\hat{\Pi}'Z'Z\hat{\Pi})^{-1}\hat{\Pi}'Z'Z_2
\end{align}
and so therefore
\begin{align}\label{eq weakhet z2mxz2a2}
    \frac{1}{n}Z_2'M_{\hat{X}}Z_2 & = \frac{1}{n}Z_2'Z_2 - \frac{1}{n}Z_2'Z\sqrt{n}\hat{\Pi} \left(\sqrt{n}\hat{\Pi}'\frac{1}{n}Z'Z\sqrt{n}\hat{\Pi}\right)^{-1} \sqrt{n}\hat{\Pi}'\frac{1}{n}Z'Z_2 \\
    &\overset{d}{\to} Q_{22} - Q_2'\tilde{\Pi}^*(\tilde{\Pi}^*Q_{ZZ}\tilde{\Pi}^*)^{-1}\tilde{\Pi}^*Q_2
\end{align}
which depends nonlinearly on the random matrix normal variable $\tilde{\Pi}^*$. If the model is estimated via LIML, which causes $\tilde{\Pi}^*$ to have a particularly nonstandard distribution, the resulting expression in (\ref{eq weakhet z2mxz2a2}) will be highly nonstandard. The specific structure of the heteroskedasticity will dictate how these nonlinearities in $\tilde{\Pi}^*$ interact, and general statements are difficult. If homoskedasticity was imposed, then e.g. $\Omega_{\tilde{Z}_{2},u} = \sigma_u^2[Q_{22} - Q_2'\tilde{\Pi}^*(\tilde{\Pi}^*Q_{ZZ}\tilde{\Pi}^*)^{-1}\tilde{\Pi}^*Q_2]$. With these components, we now state the main theoretical result of the paper:

\begin{thm}\label{thm robscore}
Let Assumptions \ref{as weakhom xziid}, \ref{as weakhet qzz}, \ref{as weakhet errors}, \ref{as weakhet pi} and \ref{as weakhet ztilde} hold. Then,
\begin{equation}\label{eq weakhet scorethm}
    S_r(\hat{\beta}) \overset{d}{\to} [\Psi^*_{2,u} - (Q_2'C + \Psi^*_{2,V})\tilde{\beta}^{*}]' \Omega_{\tilde{Z}_2\hat{u}}^{-1} [\Psi^*_{2,u} - (Q_2'C + \Psi^*_{2,V})\tilde{\beta}^{*}],
\end{equation}
where $\hat{\beta}$ is either the 2SLS or LIML estimator.
\end{thm}
All components of the expression except $Q_2'C$ are random under weak-instrument asymptotics and (\ref{eq weakhet scorethm}) is a highly non-standard distribution.  In light of our previous discussion, the above limiting distribution nests both the behaviour of the $J$- and $KP$-tests. Therefore, we have the following corollary: 

\begin{corol}\label{corol}
Let Assumptions \ref{as weakhom xziid}, \ref{as weakhet qzz}, \ref{as weakhet errors}, \ref{as weakhet pi} and \ref{as weakhet ztilde} hold. Then,
\begin{align}\label{eq weakhet hansenlimit}
        J \overset{d}{\to} [\Psi^*_{2,u}& - (Q_2'C + \Psi^*_{2,V})\tilde{\beta}_{2SLS}^{*}]' \Omega_{\tilde{Z}_2\hat{u}_{2SLS}}^{-1} [\Psi^*_{2,u} - (Q_2'C + \Psi^*_{2,V})\tilde{\beta}_{2SLS}^{*}], \\
        KP \overset{d}{\to} &[\Psi^*_{2,u} - (Q_2'C + \Psi^*_{2,V})\tilde{\beta}_L^{*}]' \Omega_{\tilde{Z}_2\hat{u}_L}^{-1} [\Psi^*_{2,u} - (Q_2'C + \Psi^*_{2,V})\tilde{\beta}_L^{*}]. \label{eq weakhet kplimit}
\end{align}
\end{corol}
Whether the $J$- or $KP$-test performs better is not obvious from inspection of the limiting distributions and will be dependent on model attributes such as the strength of endogeneity and heteroskedasticity. Assessing the performance of these tests will be the focus of the next section, where we discuss Monte Carlo results.

\section{Monte Carlo simulations}\label{section mc}

The baseline model is given by
\begin{align} \label{eq mc model}
\begin{split}
y_i &= \beta_0 + \beta_1x_i +u_i, \\
x_i &= \pi_0 + \sum_{j=1}^{k_z} \pi_j z_{j,i}+v_i.
\end{split}
\end{align}
We assume $y_i,x_i,z_{1,i},...,z_{k_z,i},u_i,v_i\in\mathbb{R}$ are scalars for $i \in\{1,...,n\}$, with $n=120$ for all simulations. For simplicity, we set $\beta_0 = \beta_1 = \pi_0 =0$. The instruments are independent standard normals and are assumed to be equally informative, such that $\pi_j = \pi = c_0/\sqrt{n}$ for all $j\in\{1,...,k_z\}$. The errors are generated according to $u_i = z_{1,i}^{\alpha} u_i^*$ and $v_i = z_{1,i}^{\alpha} v_i^*$ for $(u_i^*,v_i^*)$ i.i.d. jointly normal with unit variances and correlation coefficient $\rho$. The strength of the heteroskedasticity is determined by $\alpha\in\mathbb{R}_{+}$ and is strictly increasing in the parameter ($\alpha=0$ represents homoskedasticity). We vary $\alpha\in\{0.5,1,1.5\}$ to represent weak, medium and strong heteroskedasticities. Different degrees of endogeneity are considered with $\rho\in\{0.2,0.5,0.95\}$. By varying $\mu^2$ across the interval $[0,32]$, we can see the evolution of test statistic performance as instrument strength increases from extremely weak to strengths commonly seen in applied practice. The number of instruments (not including the constant) is set to $k_z\in\{2,4\}$, leading to 1 and 3 degrees of overidentification. These degrees of overidentification coincide with those found in the two empirical applications in Sections \ref{sec yogo} and \ref{sec pozzi}.\footnote{Simulations with $k_z=6$ and $k_z=11$ were also conducted for this design; the increases in the size distortions when increasing the degrees of overidentification (particularly for $J$ but also for $KP$ to some extent) happen in a predictable manner. The qualitative comparison of the tests is unchanged. We therefore do not report these results.}  Nominal size is calculated as the proportion of experiments for which each test exceeds the 5\% critical value for the $\chi^2(k_z-k_x)$ distribution i.e. the correct asymptotic critical value under strong identification. In Appendix \ref{app estimatorresults}, we further report the median bias and 90:10 percentile ranges for the 2SLS and LIML estimates of $\beta_1$, as well as additional size results at the 10\% and 1\% level for $J$ and $KP$ across a range of parameter configurations.

\begin{figure}[ht!] 
    \centering
    \subfloat[$k_z=2$, $\rho = 0.2$]{{\includegraphics[width=8.2cm]{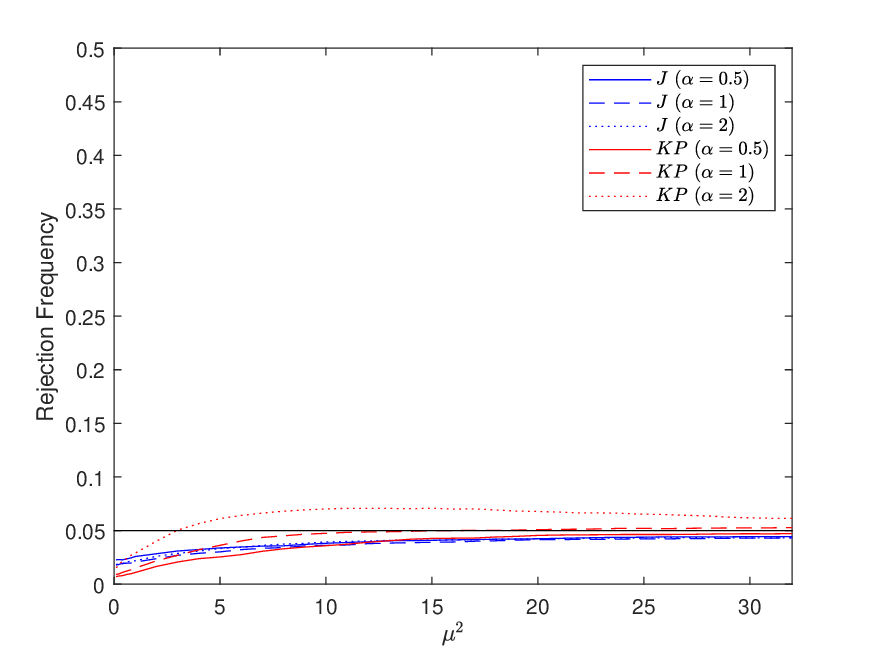} }}
    \subfloat[$k_z=4$, $\rho = 0.2$]{{\includegraphics[width=8.2cm]{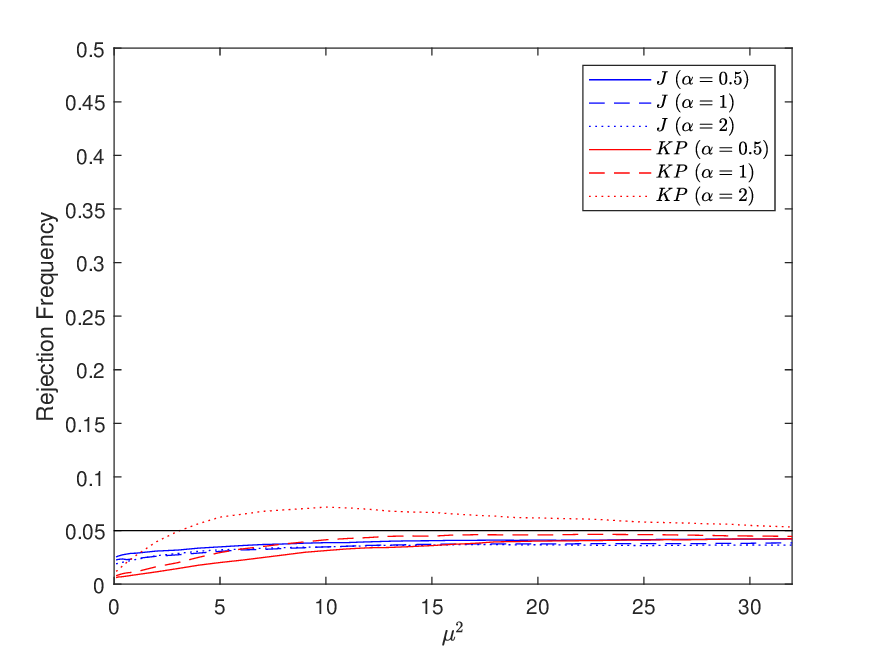} }}
    
    \centering
    \subfloat[$k_z=2$, $\rho = 0.5$]{{\includegraphics[width=8.2cm]{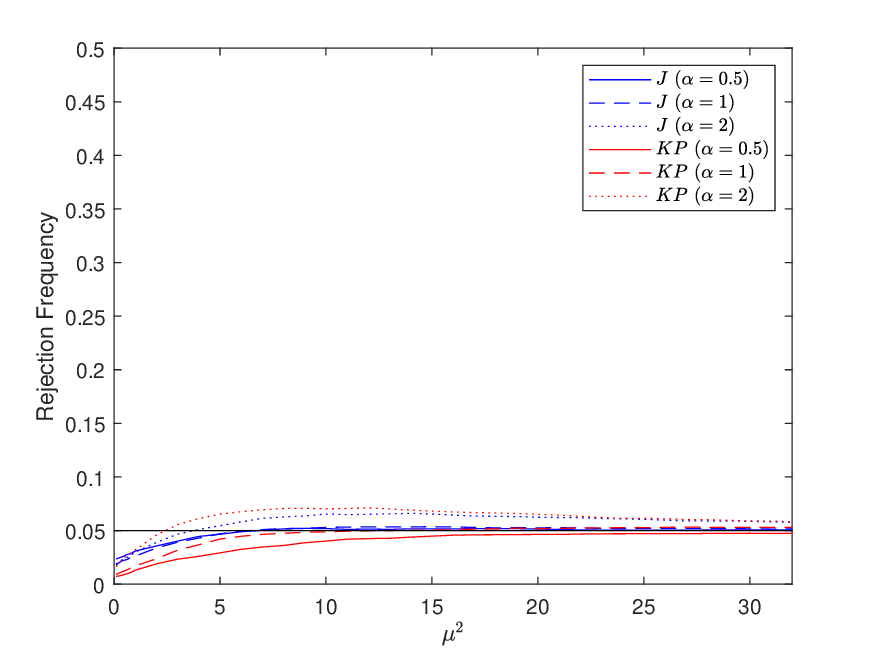} }}
    \subfloat[$k_z=4$, $\rho = 0.5$]{{\includegraphics[width=8.2cm]{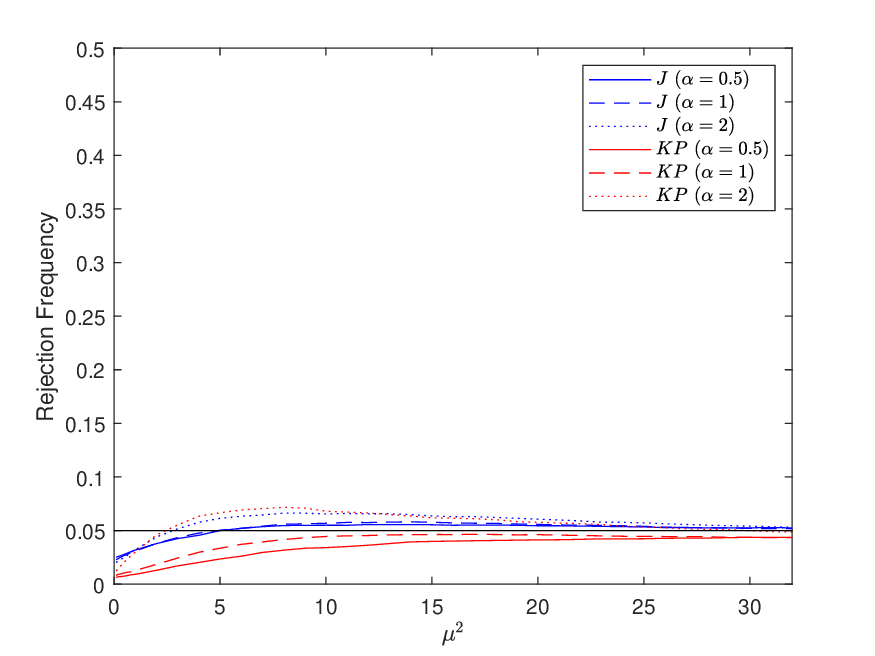} }}
    
    \centering
    \subfloat[$k_z=2$, $\rho = 0.95$]{{\includegraphics[width=8.2cm]{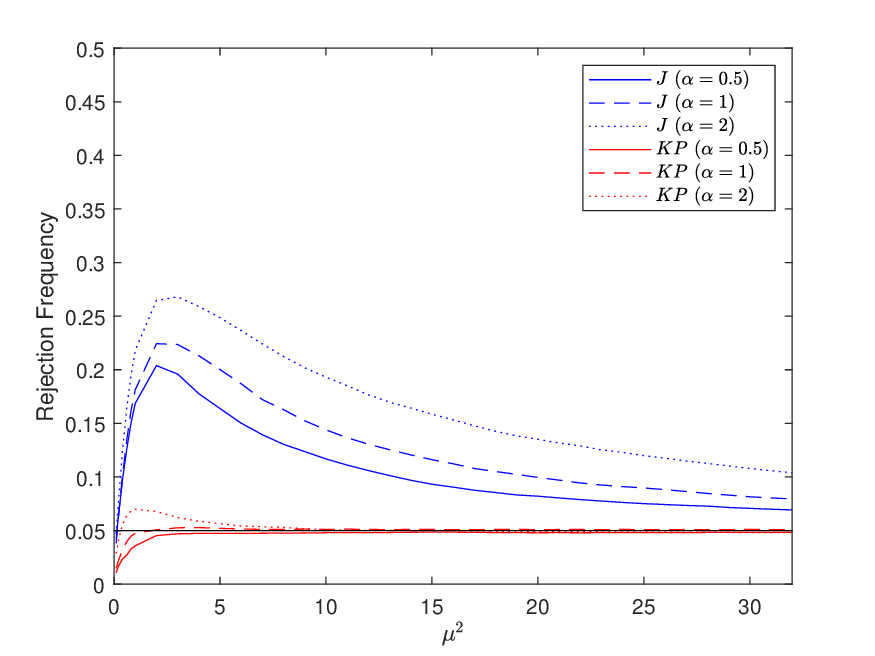} }}
    \subfloat[$k_z=4$, $\rho = 0.95$]{{\includegraphics[width=8.2cm]{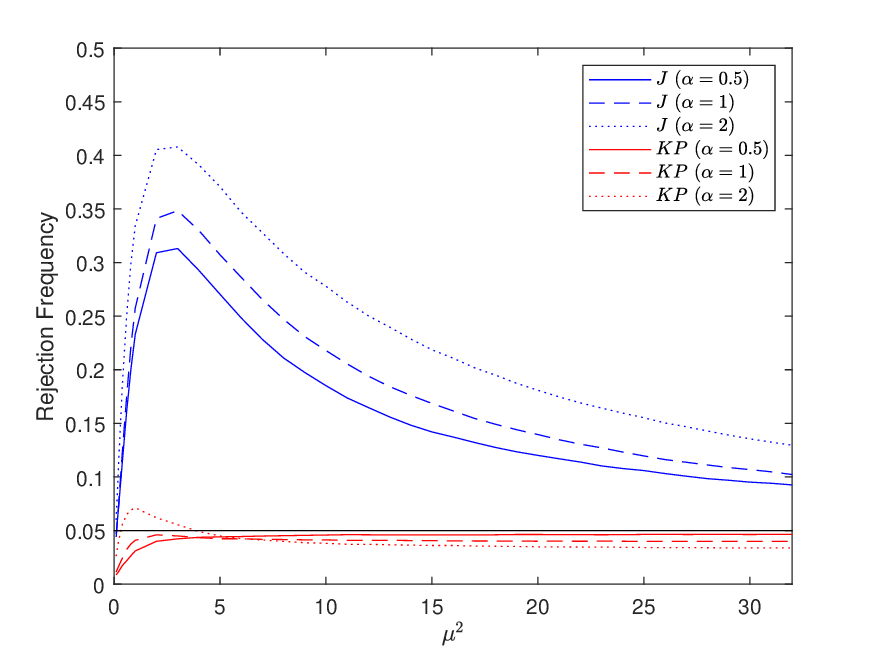} }}
    \caption[Design 1: Rejection frequency at nominal 5\% level across $\mu^2$]{Rejection frequency at nominal 5\% level.}
    \label{graph design1}
\end{figure}

Figure \ref{graph design1} presents rejection frequencies of the $J$- and $KP$-tests under the null at the nominal 5\% level. In the low endogeneity design with $k_z=2$, test performance is similar under the weak and medium strength heteroskedasticity designs. The $J$-test size is closer to the nominal level for very weak instruments, although both tests still under-reject to significant degrees. It appears as though both tests get closer to nominal size for smaller $\mu^2$ when $\rho = 0.5$ as opposed to $\rho = 0.2$. $KP$ perhaps performs marginally better when $\rho=0.2$ and $J$ performs marginally better when $\rho=0.5$, although the difference is small. The high heteroskedasticity design causes $J$ to under-reject quite severely, whereas for $KP$, the test slightly over-rejecting as $\mu^2$ increases. Performance is largely unaffected by increasing the number of instruments to $4$. In the medium endogeneity design, test performance is again fairly similar across both instrument numbers considered and across different heteroskedasticity strengths, with $J$ slightly over-rejecting and $KP$ slightly under-rejecting. All tests are undersized for $\mu^2<15$ in the three lower endogeneity designs. When $k_z=2$, the two test statistics have similar rejection frequencies for the weak and medium heteroskedasticity designs, differing only significantly in the strong heteroskedasticity design; $J$ somewhat under-rejects even at $\mu^2=32$, whereas $KP$ slightly over-rejects. 

The difference between test sizes is stark in the high endogeneity case. In the $k_z=2$ case, $KP$ performs well for $\mu^2>4$, with correct size obtained in the weak and medium heteroskedasticity designs. The test over-rejects slightly in the strongly heteroskedastic case but size distortions are relatively small. However, the $J$-test massively over-rejects, with the problem worsening with increasing strength of heteroskedasticity. In the weakest heteroskedasticity design, the size of the $J$-test peaks at $\approx0.2$ for $\mu^2=2$, before falling to $\approx 0.11$ for $\mu^2=10$. This problem gets increasingly worse as the strength of the heteroskedasticity increases; in the strong design, size peaks at $\approx 0.25$ for $\mu^2=3$, and even at $\mu^2=32$ (typically considered as relatively strong instruments), size is still at $\approx 0.09$. The problem also increases in the degrees of overidentification. For the $KP$-test, performance remains roughly consistent over the different degrees of overidentification, with the size remaining similar for weak and medium heteroskedasticity, whereas for $k_z= 4$, the over-rejection for $J$ in Panel (f) is obvious on inspection.

\begin{figure}[h!]
    \centering
    \subfloat[$k_z=2$] {{\includegraphics[width=16cm]{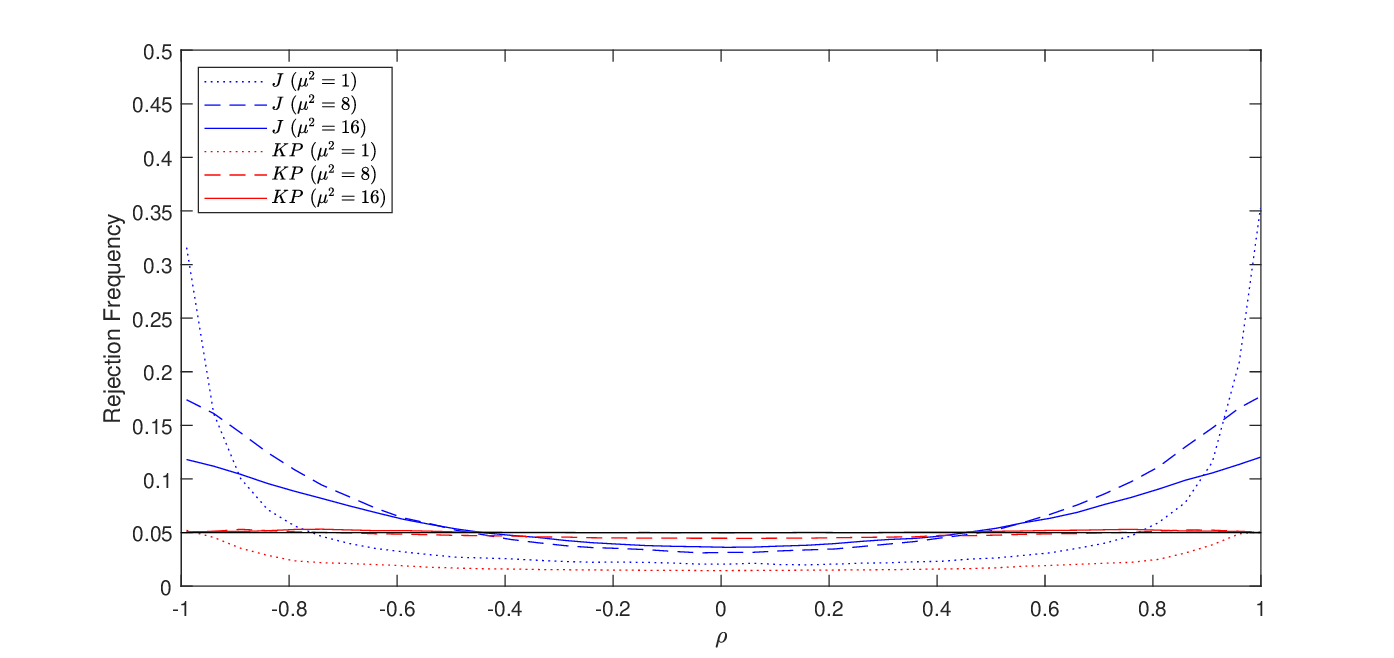}}}

    \centering
    \subfloat[$k_z=4$] {{\includegraphics[width=16cm]{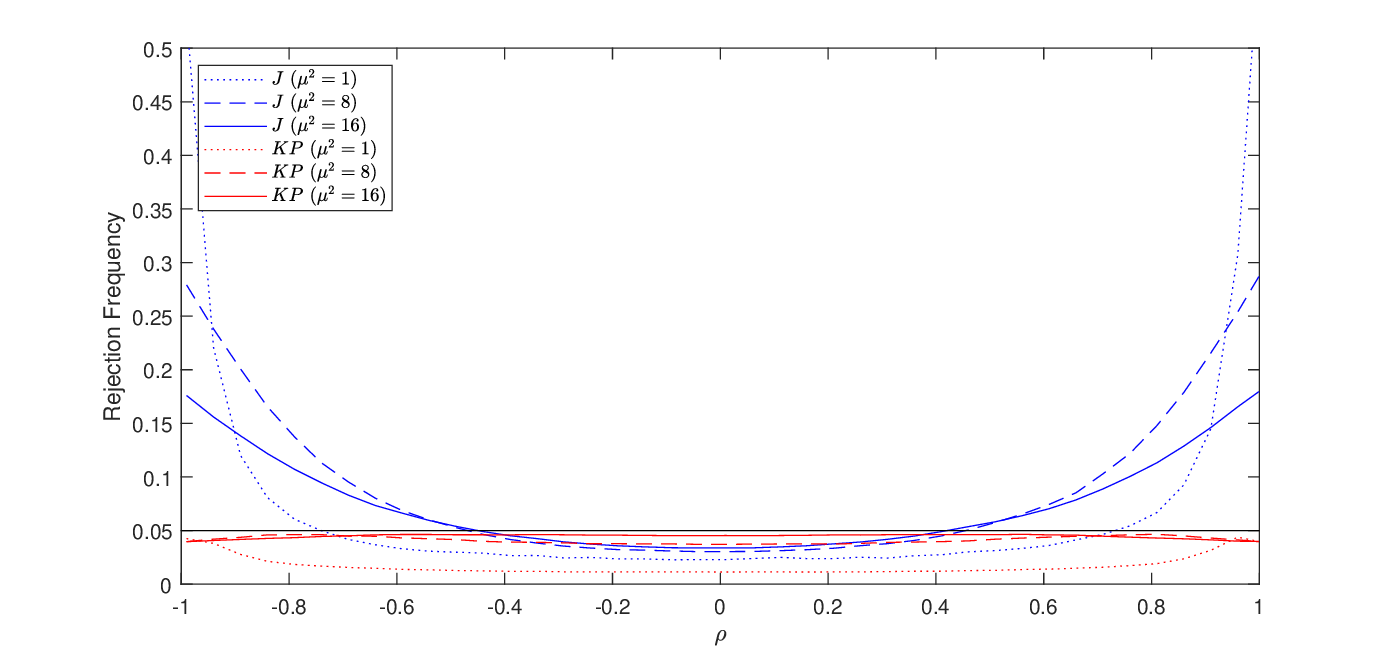}}}
    \caption[Design 1: Rejection frequency at nominal 5\% level across $\rho$]{Rejection frequency at nominal 5\% level across $\rho$.}
    \label{graph design1rho}
\end{figure}

Figure \ref{graph design1rho} shows the rejection frequencies of the test statistics for different values of $\mu^2$ for $\rho\in[-1,1]$, with $\alpha = 1$. When instruments are extremely weak for $k_z=2$, the $J$-test is closer to nominal size than $KP$ for smaller values of $|\rho|$, but for high endogeneity designs size increases sharply, tending toward 0.35, whereas $KP$ does not suffer from over-rejection. For $\mu^2=1$, performance of the two tests is similar for approximately $|\rho|\leq 0.6$, but for larger values of $|\rho|$, the $J$-test rejection frequency starts increasing sharply (at a rate increasing in $|\rho|$). The $KP$-test is clearly superior when we consider either $\mu^2 = 8$ or $\mu^2 = 16$; in the latter case, $KP$ exhibits almost perfectly correct size across all $\rho$ considered, whereas $J$ under-rejects for small $|\rho|$, and over-rejects for large $|\rho|$. A similar pattern is seen in Panel (b) for $k_z=4$, but the increase in size for the $J$-test is much more dramatic. Although the $KP$-test under-rejects more in this model than with $k_z=2$, the difference is not nearly as pronounced as the sensitivity of $J$ to the number of overidentifying restrictions. $KP$ does however have an under-rejection problem for very weak instruments. On balance, Figure \ref{graph design1rho} provides evidence in favour of $KP$ when considering the whole parameter space for $\rho$. In particular, it appears that the costs of using $J$ tend to be greater when it is chosen for a model that it is ill-suited to relative to $KP$; $J$ tends to offer only a minor improvement over $KP$ in the models where it is more favourable, whereas $KP$ is sometimes dramatically superior to $J$.

The costs of using $J$ tend to be greater when it is chosen for a model that it is ill-suited to relative to $KP$; $J$ tends to offer only a minor improvement over $KP$ in the models where it is more favourable, whereas $KP$ is sometimes dramatically superior to $J$. From e.g. Figure \ref{graph design1rho}, $KP$ has an edge over $J$ in a less weakly-identified setting when instruments are perhaps somewhat weak but not extremely weak (as is often the case in empirical work).\footnote{\textcite{hansen2008} find the median concentration parameter to be 23.6 in their study of applied microeconometric papers, and at this level of instrument strength, the $KP$-test typically outperforms the $J$-test (with this difference often large).} The structure of the heteroskedasticity and degree of endogeneity are important factors when comparing $J$ and $KP$, but these cannot be consistently estimated in a weak instrument case. This leads to the recommendation based on simulation evidence that the $KP$-test should be employed by applied researchers when testing overidentifying restrictions. 

\section{The lifecycle consumption model}\label{sec emp model}

\vspace{3mm}

In the lifecycle consumption model, an agent maximises their expected lifetime utility by smoothing consumption across an infinite time horizon. \textcite{epstein1989} and \textcite{weil1989} propose a generalised class of utility functions $U_t$, defined recursively as
\begin{equation}\label{eq emp epsteinutility}
    U_t = \left( (1-\delta)C_t^{(1-\gamma)/\varphi} + \delta\mathbb{E}_t\left[U_{t+1}^{1-\gamma}\right]^{1/\varphi}  \right)^{\varphi/(1-\gamma)},
\end{equation}
where $\mathbb{E}_t[ \, \cdot \, ]$ denotes the conditional expectation operator w.r.t. the information set at time $t$, $C_t$ is consumption at time $t$, $\delta$ is an intertemporal discount factor and $\varphi = (1 - \gamma)/(1-1/\psi)$. Here, the parameter $\psi$ represents the elasticity of intertemporal substitution and $\gamma$ is the coefficient of relative risk aversion (CRRA).\footnote{When $\vp = 1$, (\ref{eq emp epsteinutility}) collapses to the power utility function, leading to the well-known inverse relationship that the EIS is the reciprocal of the CRRA e.g. $\psi = 1/\gamma$.} Following \textcite{campbell2003} and \textcite{yogo2004}, the individual then maximises (\ref{eq emp epsteinutility}) w.r.t. the budget constraint
\begin{equation}\label{eq emp budget}
    W_{t+1} = (1 + R_{w,t+1})(W_t - C_t),
\end{equation}
where $W_{t+1}$ is total household wealth at time $t+1$ and $R_{w,t+1}$ is the gross real return on total asset investments at time $t+1$. Solving this optimisation problem yields the Euler equation
\begin{equation}\label{eq emp epsteineuler}
    \mathbb{E}_t\left[ \left( \delta \left( \frac{C_{t+1}}{C_t} \right)^{-\frac{1}{\psi}}\right)^{\varphi}\left(\frac{1}{1+ R_{w,t+1}}\right)^{1-\varphi}(1 + R_{t+1}) \right] = 1,
\end{equation}
where $R_{t+1}$ is the return on the asset in question. By imposing homoskedasticity and joint log-normality on consumption and asset returns (conditional on information available at time $t$), and log-linearising (\ref{eq emp epsteineuler}), we can derive the return on a riskless asset as
\begin{equation}\label{eq emp risklessrateepstein}
    r_{f,t+1} = -\ln \delta + \frac{1}{\psi}\mathbb{E}_t[\Delta c_{t+1}]+ \frac{\varphi-1}{2}\sigma_w^2 - \frac{\varphi}{2\psi^2}\sigma_c^2,
\end{equation}
where for the generic variable $A_{t+1}$ we denote $a_{t+1} = \ln A_{t+1}$, $\Delta a_{t+1} = a_{t+1} - a_t$ and $\sigma^2_a$ is the unconditional variance of $\{a_t\}$. Similarly, the risk-premium on risky assets is
\begin{equation}\label{eq emp riskpremiumepstein}
    \mathbb{E}_t[r_{t+1}] - r_{f,t+1} = \frac{1}{\psi}\varphi\sigma_{rc} + (1-\varphi) \sigma_{rw} - \frac{1}{2}\sigma_r^2.
\end{equation}
Now define the constant terms
$\mu_f = [-2\ln \delta + (\varphi-1)\sigma_w^2 - \varphi\sigma_c^2/\psi^2]/2$ and $\mu_r = \mu_f - \sigma_r^2/2 + \varphi\sigma_{rc}/\psi + (1-\varphi) \sigma_{rw}$.
Then, (\ref{eq emp riskpremiumepstein}) can be re-written as
\begin{equation}\label{eq prereg2}
    \mathbb{E}_t[r_{t+1}]= \mu_r + \frac{1}{\psi}\mathbb{E}_t[\Delta c_{t+1}].
\end{equation}
From (\ref{eq prereg2}), we can derive estimable regression equations. Define the error $\eta_{t+1} = r_{t+1} - \mathbb{E}_t[r_{t+1}] - \frac{1}{\psi}(\Delta c_{t+1} - \mathbb{E}_t[\Delta c_{t+1}])$. By adding $\eta_{t+1}$ to both sides, (\ref{eq prereg2}) becomes
\begin{gather}\label{eq emp reg2generaleta}
    \mathbb{E}_t[r_{t+1}] + \eta_{t+1}  = \mu_r + \frac{1}{\psi}\mathbb{E}_t[\Delta c_{t+1}] + \eta_{t+1},
\end{gather}
so substituting in the explicit formula for $\eta_{t+1}$ on the L.H.S. of (\ref{eq emp reg2generaleta}) yields
\begin{gather}
    \mathbb{E}_t[r_{t+1}] + r_{t+1} - \mathbb{E}_t[r_{t+1}] - \frac{1}{\psi}(\Delta c_{t+1} - \mathbb{E}_t[\Delta c_{t+1}]) = \mu_r + \frac{1}{\psi}\mathbb{E}_t[\Delta c_{t+1}] + \eta_{t+1} \notag \\
    \implies r_{t+1} - \frac{1}{\psi}(\Delta c_{t+1} - \mathbb{E}_t[\Delta c_{t+1}]) =  \mu_r + \frac{1}{\psi}\mathbb{E}_t[\Delta c_{t+1}] + \eta_{t+1} \notag \\
    \implies r_{t+1} = \mu_r + \frac{1}{\psi}\Delta c_{t+1} + \eta_{t+1} \label{eq emp reg1}
\end{gather}
which defines a regression equation in terms of the observed $r_{t+1}$ and $\Delta c_{t+1}$, the unobserved $\eta_{t+1}$ and the unknown constant and slope parameters $\mu_r$ and $1/\psi$ to be estimated. (\ref{eq prereg2}) can also be re-normalised in terms of $\Delta c_{t+1}$, yielding the regression equation
\begin{equation}\label{eq emp reg2}
    \Delta c_{t+1} = \mu_c + \psi r_{t+1} + u_{t+1},
\end{equation}
where $\mu_c = -\psi \mu_r$ and $u_{t+1} = -\psi\eta_{t+1} =\Delta c_{t+1} - \mathbb{E}_t[\Delta c_{t+1}] - \psi(r_{t+1} - \mathbb{E}_t[r_{t+1}])$. We consider (\ref{eq emp reg2}) to be the primary regression of interest, as the EIS is usually considered the main parameter of interest. Estimation of this parameter has received substantial attention in the literature \parencite{hansen1983,hall1988}, although a large variety of proposed plausible values have been given.

OLS is inappropriate for both regression equations, as clearly the explanatory variables will be endogenous; the errors are defined as explicit functions of the regressor. However, $u_{t+1}$ and $\eta_{t+1}$ are both errors in expectations of future variables, meaning that they are orthogonal to variables known at time $t$ \parencite{nakamura2021}. Given this, a wide range of lagged macroeconomic variables have been used as instruments in the literature. Suppose we have a $k_z$-vector $Z_{t+1}'$ of lagged macroeconomic indicators $(\, \Tilde{z}_{1,t-1} \,\,\, \Tilde{z}_{2,t-1} \,\,\, ... \,\,\, \Tilde{z}_{k_z,t-1} \, )$ as instruments (e.g. the $k^{th}$ instrument $z_{k,t+1}$ is formed by twice-lagging the $k^{th}$ macroeconomic indicator $\Tilde{z}_{k,t+1}$), stacked into the $T\times k_z$ matrix $Z$. The IV model then implies the moment restrictions $\mathbb{E}_t[Z_{t+1}\eta_{t+1}] = 0$ and $\mathbb{E}_t[Z_{t+1}u_{t+1}] = 0$. From the structure of the error terms $\eta_{t+1}$ and $u_{t+1}$, it is clear that the restrictions $\mathbb{E}_t[Z_{t+1}\eta_{t+1}] = 0$ and $\mathbb{E}_t[Z_{t+1}u_{t+1}] = 0$ are equivalent up to linear transformations \parencite{yogo2004}.

Heteroskedasticity is an important feature of the model, as it has the interpretation of representing precautionary savings. Under heteroskedasticity, \textcite{yogo2004} shows that the covariance and variance terms that enter the constants in the regressions simply need be replaced by their conditional counterparts. By appropriately redefining the error terms $u_{t+1}$ and $\eta_{t+1}$, we can arrive at the moment restrictions to identify $\psi$ and $1/\psi$. To see this, consider the risk-free log-linearised Euler equation $r_{f,t+1} = \mu_f + \frac{1}{\psi}\mathbb{E}_t[\Delta c_{t+1}]$ for simplicity, but allow for time-varying conditional variances such that
\begin{equation}
    \mu_{f,t} = -\ln \delta + \frac{\varphi-1}{2}\sigma_{w,t}^2 - \frac{\varphi}{2\psi^2}\sigma_{c,t}^2.
\end{equation}
Then, \textcite{yogo2004} shows that as long as $\mathbb{E}_t[Z_t(\mu_{f,t}-\mu_f)] = 0$, then $1/\psi$ is still identified by the moment condition $\mathbb{E}_t[Z_{t+1}\eta_{t+1}]=0$, with an equivalent result available for the moment condition $\mathbb{E}_t[Z_{t+1}u_{t+1}]=0$ identifying $\psi$.

To test the validity of the moment restrictions implied by (\ref{eq emp reg1}), we want to test
\begin{equation}
    \mathbb{H}_0:\mathbb{E}_t[Z_{t+1}\eta_{t+1}] = 0  \textup{ v.s. } \mathbb{H}_1:\mathbb{E}_t[Z_{t+1}\eta_{t+1}] \neq 0, \label{eq emp hyp2}
\end{equation}
and for (\ref{eq emp reg2}), we want to test
\begin{equation}
    \mathbb{H}_0:\mathbb{E}_t[Z_{t+1}u_{t+1}] = 0  \textup{ v.s. } \mathbb{H}_1:\mathbb{E}_t[Z_{t+1}u_{t+1}] \neq 0. \label{eq emp hyp1}
\end{equation}
which is possible as long as $k_z > 1$. Whilst standard practice for testing (\ref{eq emp hyp2}) and (\ref{eq emp hyp1}) under strong identification is well-established in the literature, we note that weak instruments are a well-known problem in the estimation of the EIS \parencite{yogo2004,neely2001}, as are questions about the validity of the instruments, which stem from rejections of the overidentifying restrictions from $J$-tests \parencite{epstein1991,gomes2011,gomes2013,dacy2011,pakos2011}. Given this and the empirical relevance of heteroskedasticity in our model, we suggest that standard practices regarding overidentification testing may not be suitable for this application.

\section{Application I: \textcite{yogo2004}} \label{sec yogo}

The dataset used for this application comes from \textcite{yogo2004}, and is commonly used as an empirical application in the weak instrument literature \parencite{montiel2013,andrews2016}. This dataset consists of quarterly data on stock markets at the aggregate level, as well as macroeconomic variables from 11 countries: Australia (AUS), Canada (CAN), France (FRA), Germany (GER), Italy (ITA), Japan (JAP), Netherlands (NTH), Sweden (SWD), Switzerland (SWT), the United Kingdom (UK) and the United States of America (USA). The stock market data come from Morgan Stanley Capital International, and the consumption and interest rate data come from the International Financial Statistics of the International Monetary Fund. For the USA, consumption of nondurables and services is measured in the dataset, but for every other country, only data on total consumption are available. Due to data constraints, the time periods vary across countries, with most time series beginning in 1970, although the data for the USA stretches back to 1947. On top of the consumption and interest rate data, the dataset contains variables for a vector of instruments, consisting of nominal interest rate, inflation, consumption growth and log dividend-price ratio, all of which are twice-lagged. See \textcite{yogo2004} for a more complete description of the dataset.

\subsection{Results}\label{subsec yogoresults}

Table \ref{table empirical} reports estimation and test statistic results from the 11 countries. Panel (a) reports results using model (\ref{eq emp reg2}) i.e. where the EIS $\psi$ is the parameter of interest, and Panel (b) reports results using model (\ref{eq emp reg1}), where the reciprocal of the EIS $1/\psi$ is the parameter of interest. Column 3 presents the $F_{eff}$-statistic from \textcite{montiel2013}, the current recommended statistic for testing weak instruments under non-homoskedastic errors. $\kappa$ is the simplified, conservative 95\% critical value - given the effects that weak instruments have on estimation and inference, we choose the strongest critical value from \textcite{montiel2013} as the benchmark to rule out weak instruments. Columns 5-6 reports estimates of $\psi$ in Panel (a) and $1/\psi$ in Panel (b), estimated via 2SLS and LIML respectively. Columns 7-8 report the $J$ and $KP$ test statistic values for hypotheses (\ref{eq emp hyp1}) and (\ref{eq emp hyp2}). For convenience, we highlight in bold $F_{eff}$-statistics that fail to reject the null of weak instruments, and values of $J$ and $KP$ that are higher than the critical value $\chi^2_{0.95}(3) = 7.815$.

\vspace{2mm}

\addtolength{\tabcolsep}{6pt} 
\begin{table}[ht!]
\small
\begin{center}
\begin{threeparttable}
\caption{Estimates using dataset from \textcite{yogo2004}.} 
\vspace{-0.5em}
\centering 
\begin{tabular*}{\textwidth}{@{\extracolsep{\fill}}  c c c c c c c c}
\hline\hline
\multicolumn{8}{c}{Panel (a): $\Delta c_{t+1} = \mu_c + \psi r_{t+1} + u_{t+1}$} \T \\\cline{2-3}
\hline  
Country & Year-Quarter &  $F_{eff}$ & $\kappa$ & $2SLS$ & $LIML$ & $J$ & $KP$ \T \\ 
\hline
\hline
AUS & 1970.Q3-1998.Q4 & 19.18 & 18.40 & 0.05 & 0.03 & \textbf{8.78} & \textbf{8.89} \\
CAN & 1970.Q3-1999.Q1 & \textbf{13.86} & 18.58 & -0.30 & -0.34 & 5.04 & 5.05 \\
FRA & 1970.Q3-1998.Q3 & 41.97 & 19.31 & -0.08 & -0.08 & 0.45 & 0.45 \\
GER & 1979.Q1-1998.Q3 & \textbf{13.37} & 18.32 & -0.42 & -0.44 & 2.59 & 2.54 \\
ITA & 1971.Q4-1998.Q1 & 21.44 & 18.92 & -0.07 & -0.07 & 1.07 & 1.06 \\
JAP & 1970.Q3-1998.Q4 & \textbf{5.43} & 21.29 & -0.04 & -0.05 & 4.73 & 4.73 \\
NTH & 1977.Q3-1998.Q4 & \textbf{12.18} & 18.53 & -0.15 & -0.14 & 3.69 & 3.69 \\
SWD & 1970.Q3-1999.Q2 & 21.19 & 18.76 & -0.00 & -0.00 & 2.59 & 2.59 \\
SWT & 1976.Q2-1998.Q4 & \textbf{7.90} & 18.03 & -0.49 & -0.50 & 2.25 & 2.27 \\
UK  & 1970.Q3-1999.Q1 & \textbf{8.44} & 20.11 & 0.17 & 0.16 & 5.05 & 5.07 \\
USA & 1947.Q3-1998.Q4 & \textbf{8.14} & 18.21 & 0.06 & 0.03 & 7.14 & 7.58 \\  
\hline
\multicolumn{8}{c}{Panel (b): $  r_{t+1} = \mu_r + (1/\psi)\Delta c_{t+1} + \eta_{t+1}$} \T \\\cline{2-3}
\hline 
Country & Year-Quarter & $F_{eff}$ & $\kappa$ & $2SLS$ & $LIML$ & $J$ & $KP$ \T \\ 
\hline
\hline
AUS & 1970.Q3-1998.Q4 & \textbf{2.47} & 19.50 & 0.50 & 30.03 & \textbf{9.49} & \textbf{8.89} \\
CAN & 1970.Q3-1999.Q1 & \textbf{2.98} & 18.07 & -1.04 & -2.98 & 6.96 & 5.05 \\
FRA & 1970.Q3-1998.Q3 & \textbf{0.22} & 19.67 & -3.12 & -12.38 & 2.07 & 0.45 \\
GER & 1979.Q1-1998.Q3 & \textbf{1.13} & 18.59 & -1.05 & -2.29 & 3.16 & 2.54 \\
ITA & 1971.Q4-1998.Q1 & \textbf{0.49} & 18.90 & -3.34 & -14.81 & 3.99 & 1.06 \\
JAP & 1970.Q3-1998.Q4 & \textbf{1.98} & 17.89 & -0.18 & -21.56 & \textbf{8.42} & 4.73 \\
NTH & 1977.Q3-1998.Q4 & \textbf{1.67} & 19.16 & -0.53 & -6.94 & \textbf{9.91} & 3.69 \\
SWD & 1970.Q3-1999.Q2 & \textbf{0.87} & 17.28 & -0.10 & -399.86 & \textbf{13.28} & 2.59 \\
SWT & 1976.Q2-1998.Q4 & \textbf{1.58} & 19.85 & -1.56 & -2.00 & 2.92 & 2.27 \\
UK  & 1970.Q3-1999.Q1 & \textbf{2.68} & 17.63 & 1.06 & 6.21 & \textbf{8.17} & 5.07 \\
USA & 1947.Q3-1998.Q4 & \textbf{2.65} & 17.61 & 0.68 & 34.11 & \textbf{9.84} & 7.58

 \\  
\hline 
\end{tabular*}
\vspace{1mm}
\caption*{Panels (a) and (b) give estimates/test statistic values for models  $\Delta c_{t+1} = \mu_c + \psi r_{t+1} + u_{t+1}$ and $r_{t+1} = \mu_r + (1/\psi) \Delta c_{t+1} + \eta_{t+1}$ respectively. Column 3 gives the $F_{eff}$-statistic from \textcite{montiel2013}. Column 4 gives conservative simplified 95\% critical value for $F_{eff}$. Columns 5-6 give the estimates of $\psi$ in (a) and $1/\psi$ in (b) by 2SLS and LIML respectively. Columns 7-8 report the $J$- and $KP$-test statistic values for the test of overidentification. Values in bold in Column 3 suggest we cannot reject the null of weak instruments, and bold values in Columns 7-8 exceed the 5\% significance critical value $\chi^2_{0.95}(3) = 7.815$, so the null hypothesis of (\ref{eq emp hyp1}) for Panel (a) or (\ref{eq emp hyp2}) for Panel (b) is rejected. Newey-West standard errors are used with $L = 6$ lags for the USA and $L=4$ lags for all other countries (the difference accounts for the longer time horizon of the USA data.)
}
\label{table empirical}
\end{threeparttable}
\end{center}
\end{table}

In Panel (a), we immediately see that both $J$ and $KP$ reject the null of exogeneity for Australia, and fail to reject the null for all other countries. There is some evidence of weak instruments, with 7 of the 11 countries having an $F_{eff}$-statistic lower than the respective conservative 95\% critical values, with many of the other specifications attaining $F_{eff}$-statistics only slightly above the critical value. Only the data from France give an $F_{eff}$-statistic substantially over the critical value. It is worth however noting that we are requiring the strongest evidence by using the conservative threshold; if we use the LIML critical values, only 4 specifications would fail to reject the null of weak instruments at the conventional level. However, 2SLS and LIML provide similar point estimates, and the test statistics give similar values in each case too, which typically suggests that identification is not particularly weak, in spite of the $F_{eff}$ values. It appears that both the $J$- and $KP$-tests still perform reasonably well on the whole; as an informal back-of-the-envelope calculation, assuming correctly-sized tests and the null holding in all 11 specifications (with independence between specifications), the probability of receiving at least one false rejection at the 5\% level is 43\%. This suggests that the single rejection for each test can be reasonably explained as a false rejection. We conclude that Panel (a) therefore fails to provide strong evidence against instrument validity and/or correct specification.

Panel (b) gives a different picture. Every country fails to provide evidence against a large weak-instrument problem (the largest $F_{eff}$-statistic is less than 3). The $J$-test rejects in 6 of the 11 specifications, despite the fact that we have good reason to believe that the instruments are exogenous by construction. Perhaps more importantly, we know that when identification is not weak, $J$ and $KP$ are typically numerically similar and have good size properties. Therefore, given the results from Panel (a) and knowing that that the moment restriction holding in (\ref{eq emp hyp1}) implies it holds in (\ref{eq emp hyp2}), then Panel (b) gives strong evidence of the over-rejection properties found under weak instruments in the Monte Carlo simulations of Section \ref{section mc}. This is also suggestive of the fact that the degree of endogeneity is likely higher when $\Delta c_{t+1}$ is treated as the endogenous regressor as opposed to $r_{t+1}$.\footnote{ Given the small sample size combined with the weak instrument problem, it is impossible to obtain estimates of the degree of endogeneity that could be considered reliable.} When we look at $KP$, we immediately see the numerical equivalence between the values in Panel (a) and Panel (b) for each country, which follows from the invariance to normalisation of LIML and $KP$.\footnote{ Also note that the LIML estimates are exact reciprocals, whereas this is not the case for 2SLS.} From this, it is clear that in both specifications, the $KP$-test will give the same conclusion for a particular country regardless of whether we have $\psi$ or $1/\psi$ as our unknown parameter, which is potentially a useful property when we consider that the null in (\ref{eq emp hyp1}) implies the null in (\ref{eq emp hyp2}) and vice versa. Under the assumption of power utility functions, this invariance could be advantageous in helping us obtain more reliable estimates of the CRRA (as under power utility it is the inverse of the EIS), despite the significant weak-instrument problem present. We find that $KP$ does not suffer from the same over-rejection problem in Panel (b), and can be considered more reliable for inference than the $J$ counterpart. Given the results of Panel (a), Panel (b) suggests that $J$ is over-rejecting rather than $KP$ falsely failing to reject.

Our results provide a different explanation to the small-scale replication of \textcite{gomes2011}, who replicate \textcite{yogo2004} but report homoskedastic Sargan tests. They reject the null hypothesis at the 5\% level for 4 countries and 10\% level for a further two. Given this, they raise doubts about the validity of the instruments and therefore the model specification. However, they do not consider the effects of non-homoskedastic errors or discuss the effects of weak instruments on the size of 2SLS-based overidentification tests. Our results suggest that when appropriately accounting for these two factors, we fail to find evidence against the overidentifying restrictions. The recommendation from this empirical exercise is to $KP$ for overidentification testing, and that instrument validity and/or misspecification of the moment conditions are unlikely to be the cause of the variation of proposed estimates for $\psi$ seen in the literature.

\section{Conclusion}

In this paper we have studied the $KP$-test as a test for overidentifying restrictions and provided evidence for its usefulness relative to the standard $J$-test commonly employed in econometric packages. We have derived the limiting distribution of the robust score test, nesting the $J$- and $KP$-tests, under heteroskedastic weak instruments, and find in general that the $KP$-test is favourable, especially when the degree of overidentification is small. This conclusion follows for multiple reasons: firstly, the size distortions for $KP$ tend to be less severe than for $J$ in the models where each performs poorly relative to other e.g. the extreme over-rejection of $J$ in high endogeneity models. Although $J$ seems to perform slightly better under extremely weak identification, both tests have highly nonstandard distributions and neither can be considered reliable here. Under moderately weak identification at levels commonly seen in the applied literature, $KP$ outperforms $J$ across a range of models. Secondly, $KP$ seems to be much less dependent on the strength of the heteroskedasticity; especially when $\rho$ is high, $J$ varies considerably as the strength of the heteroskedasticity changes (implying that this is important information for a researcher to consider about their model), and since the heteroskedastic error structure cannot be consistently estimated under weak instruments, this is problematic. The $KP$-test can therefore be seen as more ``robust" to variation in endogeneity or heteroskedasticity strength under weak identification. We therefore recommend its use for applied researchers. This recommendation also provides a generalisation of the guidance from \textcite{staiger1997}, who recommend using LIML-based overidentification tests under homoskedasticity.

Further, we have empirically assessed the $J$- and $KP$-tests and re-examined the validity of the over-identifying restrictions in the estimation of the EIS in the lifecycle consumption model, where weak instruments and heteroskedasticity are common issues. As expected, we provide evidence that the $J$- and $KP$-tests are extremely similar in performance when instruments are strong, with both tests failing to find evidence to reject the overidentifying restrictions. However, when the instruments are possibly weak, then $J$ is empirically much more likely to reject these restrictions than $KP$. With this in mind, we suggest that a large number of papers are possibly rejecting the overidentifying restrictions erroneously due to the poor properties of the $J$-test under weak instruments, rather than finding evidence that the instruments are invalid. Given the previous theoretical results and simulations, as well as the natural exogeneity expected of instruments formed from lagged macroeconomic variables in this model, we argue that this is evidence of the $J$-test having a substantially inflated size under weak instruments, a problem we find not present in the $KP$-test (with this conclusion supported further by the additional application in the Appendix). We therefore suggest two key implications: from an econometric perspective, we provide empirical evidence of the suitability of the $KP$-test as an overidentification test over $J$-test in situations with weak identification and heteroskedastic errors. From a macroeconomic perspective, we suggest that previous evidence found in the literature of possible instrument invalidity stems from issues with the $J$-test, rather than due to the instruments and specification of the model commonly used by empirical macroeconomists. Therefore, instrument validity does not seem to be a plausible factor in the large range of estimates of the EIS presented in the literature.   

\printbibliography

\newpage

\appendix

\section*{\LARGE{Appendix}}

\section{Empirical application II: \textcite{pozzi2022}}\label{sec pozzi}

\textcite{pozzi2022} assesses the EIS using real returns on housing. The type of asset studied in the literature has received relatively little attention, and there has been little evidence of using house market returns in empirical studies of the EIS, despite the importance of housing as an asset e.g. the 2019 Survey of Consumer Finances \parencite{FederalReserveBoard2019} reports that 65\% of households in the USA own a home as a primary residence. On the other hand, the same survey reports that only 24\% of households own stocks or bonds. \textcite{havranek2015} suggest that estimates of $\psi$ and $1/\psi$ based on stock market data may be less reliable as households who own stocks may be more willing to substitute their consumption intertemporally e.g. multiple studies such as \textcite{mankiw1991} and \textcite{vissingjorgensen2002} provide evidence that estimates of $\psi$ are larger for households who own stocks as opposed to not. \textcite{pozzi2022} therefore argues that real housing returns are likely much more important in determining consumption than stock returns.

We focus on simple per-country specifications analogous to (\ref{eq emp reg2}) of the form
\begin{equation}\label{eq emp simplepozzi}
    \Delta c_{t+1} = \mu_c + \psi r_{t+1} + u_{t+1},
\end{equation}
where $r_{t+1}$ now denotes real returns on housing.\footnote{\textcite{pozzi2022} presents a more sophisticated theoretical model than the basic model considered in Section \ref{sec emp model}. The theoretical framework presented considers a heterogeneous agent model with consumers who face time-varying preference shifters and incomplete financial markets. \textcite{pozzi2022} derives the Euler equations for the agents in the model, and provides an estimable regression equation based on the aggregation of these Euler equations. As our focus is on how the considered overidentification tests behave under weak instruments, we do not describe the theoretical model here, but rather refer readers to the paper for a full description of the model and also how it leads to the estimable regression equation (\ref{eq emp simplepozzi}).}\footnote{\textcite{pozzi2022} does not actually report any per-country simple specification results in the paper, but rather presents a group-mean panel approach. As our concern is on testing for overidentifying restrictions with weak instruments, we focus on the simple per-country specifications, analogous to the \textcite{yogo2004} application.} \textcite{pozzi2022} uses two instruments in his specifications. The first instrument is the first lag of the real housing return for each country, denoted $r_t$. \textcite{pozzi2022} notes that there are potential issues with the validity of this instrument, as any persistent omitted variables in the error term may potentially be correlated with the instrument and thus violate the exogeneity assumption, in line with \textcite{hall1988}. To allow for the testing of overidentifying restrictions to test for validity, he proposes as a second instrument the mean of international real housing returns excluding the country of interest, denoted $\bar{r}_{t+1}^{for}$. 
We run per-country regressions of the form (\ref{eq emp simplepozzi}) using the instrument vector of \textcite{pozzi2022} and test the overidentifying restrictions using $J$ and $KP$, but also run a second set of regressions using $r_{t-1}$ instead of $r_{t}$. The benefit of this is that $r_{t-1}$ is plausibly more likely to be exogenous than $r_{t}$, but is also likely to be a weaker instrument. Therefore, in this application, we run two specifications for each country: a first where the instruments are stronger but less likely valid, and a second where they are weaker but more likely valid.

\subsection{Data}

The dataset used by \textcite{pozzi2022} consists of annual housing data from 15 countries during the period 1950-2015. This dataset is itself taken from the extensive dataset constructed using the Jord\`a-Schularick-Taylor macro-history database \parencite{jorda2019}. This dataset consists of annual data on real housing returns and per capita log of real consumption as well as macroeconomic variables.\footnote{ Details of the construction of this dataset can be found in the supplementary material to \textcite{jorda2019}.} The 15 countries are: Australia (AUS), Denmark (DEN), Finland (FIN), France (FRA), Germany (GER), Italy (ITA), Japan (JAP), Netherlands (NTH), Norway (NOR), Portugal (POR), Spain (SPA), Sweden (SWD), Switzerland (SWT), the United Kingdom (UK) and the United States of America (USA). 

\subsection{Results}

\begin{table}[ht!]
\small
\begin{center}
\addtolength{\tabcolsep}{7.2pt}
\begin{threeparttable}
\caption{Estimates using dataset from \textcite{pozzi2022}} 
\vspace{-0.5em}
\centering 
\begin{tabular*}{\textwidth}{@{\extracolsep{\fill}}  c c c c c c c c}
\hline\hline
\multicolumn{7}{c}{Panel (a): Instrument set $Z_{t+1}' = (\,  r_{t},\, \bar{r}_{t+1}^{for} \,)$} \T \\\cline{2-3}
\hline
Country & $F_{eff}$ & $\kappa$ & $2SLS$ & $LIML$ & $J$ & $KP$ \\ 
\hline 
AUS & \textbf{3.48} & 20.05 & 0.20 & 0.22 & 0.56 & 0.49 \\
DEN & \textbf{9.20} & 20.97 & 0.22 & 0.22 & 0.31 & 0.31 \\
FIN & \textbf{2.91} & 21.47 & 0.40 & 0.40 & 0.03 & 0.03 \\
FRA & 119.62 & 20.34 & 0.14 & 0.14 & \textbf{4.35} & \textbf{4.36} \\
GER & 35.03 & 20.87 & 0.10 & 0.10 & 1.03 & 1.01 \\
ITA & \textbf{0.00} & 22.20 & -10.05 & -38.75 & 0.02 & 0.00 \\
JAP & 24.95 & 20.61 & 0.30 & 0.30 & 1.02 & 1.03 \\
NTH & 33.99 & 20.96 & 0.23 & 0.23 & 1.37 & 1.30 \\
NOR & \textbf{0.23} & 22.15 & 0.03 & 0.00 & 0.36 & 0.28 \\
POR & \textbf{6.85} & 21.97 & 0.43 & 0.44 & 0.33 & 0.34 \\
SPA & \textbf{2.96} & 22.34 & 0.31 & 0.31 & 0.06 & 0.06 \\
SWD & \textbf{5.96} & 19.64 & 0.25 & 0.29 & \textbf{4.25} & 0.37 \\
SWT & \textbf{8.62} & 20.74 & 0.04 & -1.25 & \textbf{6.32} & \textbf{3.85} \\
UK & 22.52 & 21.31 & 0.17 & 0.16 & \textbf{4.56} & \textbf{4.74} \\
USA & \textbf{16.42} & 20.34 & 0.23 & 0.24 & \textbf{7.15} & \textbf{7.15} \\
\hline 
\hline
\multicolumn{7}{c}{Panel (b): Instrument set $Z_{t+1}' = (\,  r_{t-1},\, \bar{r}_{t+1}^{for} \,)$} \T \\\cline{2-3}
\hline 
Country & $F_{eff}$ & $\kappa$ & $2SLS$ & $LIML$ & $J$ & $KP$ \\ 
\hline 
AUS & \textbf{4.48} & 21.28 & 0.35 & 0.37 & 1.10 & 1.13 \\
DEN & \textbf{2.63} & 21.95 & 0.17 & 0.15 & 2.43 & 2.31 \\
FIN & \textbf{4.27} & 21.60 & 0.40 & 0.40 & 0.02 & 0.02 \\
FRA & 48.50 & 20.31 & 0.16 & 0.17 & 2.79 & 2.61 \\
GER & \textbf{15.42} & 19.59 & 0.12 & 0.11 & 1.05 & 1.03 \\
ITA & \textbf{1.94} & 21.49 & 0.12 & 9.35 & \textbf{5.77} & 3.32 \\
JAP & \textbf{6.82} & 21.45 & 0.46 & 0.46 & 0.00 & 0.00 \\
NTH & \textbf{3.61} & 20.97 & 0.33 & 0.34 & 0.35 & 0.36 \\
NOR & \textbf{0.49} & 20.96 & 0.05 & -1.19 & 1.44 & 0.65 \\
POR & \textbf{3.84} & 22.04 & 0.48 & 0.48 & 0.19 & 0.19 \\
SPA & \textbf{14.20} & 20.13 & 0.32 & 0.33 & 0.18 & 0.17 \\
SWD & \textbf{3.32} & 20.70 & 0.35 & 0.35 & 0.02 & 0.02 \\
SWT & \textbf{0.83} & 21.32 & 0.26 & 2.92 & \textbf{4.57} & 2.09 \\
UK & \textbf{10.39} & 21.25 & 0.23 & 0.23 & 2.67 & 2.66 \\
USA & \textbf{2.23} & 20.95 & 0.59 & 0.74 & 1.90 & 1.38 \\
\hline 
\hline 
\end{tabular*}
\vspace{1mm}
\caption*{Data for each country spans 1950-2015 ($T=62$). Panels (a) and (b) give estimates/test statistic values for models using the instrument sets $Z_{t+1}' = (\,  r_{t},\, \bar{r}_{t+1}^{for} \,)$ and $Z_{t+1}' = (\,  r_{t-1},\, \bar{r}_{t+1}^{for} \,)$ respectively. Column 3 gives the $F_{eff}$-statistic from \textcite{montiel2013}. Column 4 gives conservative simplified 95\% critical value for $F_{eff}$.  Columns 5-6 give the estimates of $\psi$ by 2SLS and LIML respectively. Columns 7-8 report the $J$- and $KP$-test statistic values for the test of overidentification. Values in bold in Column 3 suggest we cannot reject the null of weak instruments, and bold values in Columns 7-8 exceed the 5\% significance critical value $\chi^2_{0.95}(1) = 3.841$. Newey-West standard errors are calculated with $L = 4$ lags.
}
\label{table empiricalpozzi}
\end{threeparttable}
\end{center}
\end{table}

This subsection reports 2SLS and LIML estimates of $\psi$, and the $J$- and $KP$-statistics for testing overidentifying restrictions for the 15 countries considered in the \textcite{pozzi2022} dataset. Panel (a) employs the first instrument set $Z_{t+1}' = (\,  r_{t},\, \bar{r}_{t+1}^{for} \,)$ and Panel (b) employs $Z_{t+1}' = (\,  r_{t-1},\, \bar{r}_{t+1}^{for} \,)$. The individual columns of Table \ref{table empiricalpozzi} have the same descriptions as the columns presented in Table \ref{table empirical}.

In Panel (a), two-thirds of the regressions considered might suffer from a weak instrument problem, with the $F_{eff}$-statistic failing to reach the conservative 5\% critical value $\kappa$. However, weak instruments may not be too much of a problem here, as except for Italy and Switzerland, both 2SLS and LIML provide similar estimates for $\psi$, which are typically in the range of 0.1-0.5. The values of 2SLS and LIML, and $J$ and $KP$ are often quite different when weak identification is a significant problem. Regarding the overidentifying restrictions, both tests are again similar. The two agree on whether to reject or fail to reject the null for all but one country (Sweden), where the $J$-test rejects the overidentifying restriction, but the $KP$-test does not. Since this specification uses only first-lagged instruments rather than twice-lagged instruments, it seems plausible that both tests are picking up legitimate rejections of the overidentifying restriction here.

Panel (b) presents results using the second set of instruments $Z_{t+1}' = (\,  r_{t-1},\, \bar{r}_{t+1}^{for} \,)$, and gives a somewhat similar picture to Panel (b) of Table \ref{table empirical}. Given the $F_{eff}$-statistics and their respective critical values, the null of weak instruments cannot be rejected at conventional significance levels in 14 of the 15 regressions. However, despite this weakness, 2SLS and LIML both tend to give similar point estimates, with large variation in only 3 of the specifications. For testing the overidentifying restriction, we see that $J$ rejects the null on three occasions, whereas $KP$ does not reject the null on any occasion. Given the twice-lagged nature of $r_{t-1}$, exogeneity is expected to hold in this case, much as in the case for \textcite{yogo2004}. However, we of course suffer from a serious weak instrument problem in this specification. It is likely here that we have some spurious rejections from the $J$-test, although the frequency of likely Type I errors is lower than in Panel (b) of Table \ref{table empirical}. This aligns with our simulation results; this application only has one overidentifying restriction as opposed to three in \textcite{yogo2004}. Our previous simulations suggest the scale of the over-rejection problem for $J$ increases quickly in the number of overidentifying restrictions, and that both tests tend to perform reasonably well with just a single overidentifying instrument. Regardless, we expect here that the $KP$-test is correctly failing to reject the null, but the $J$-test sees some over-rejection of the overidentifying restrictions.

\section{Proofs for Section \ref{section weakhet}}\label{section proofs}

\begin{proof}[Proof of Lemma \ref{lem weakhet staigerstock}]
These results follow similarly to those in Staiger \& Stock (1997). \\
     (\textit{a}) $Z'X/\sqrt{n} = Z'(ZC/\sqrt{n} + V)/\sqrt{n} = (Z'Z/n)C + Z'V/\sqrt{n}\overset{d}{\to} Q_{ZZ}C + \Psi^*_{ZV}$. \\
    (\textit{b}) We have
    \begin{align*}
        X'P_ZX &= (Z'X/\sqrt{n})'(Z'Z/n)^{-1}(Z'X/\sqrt{n}) \overset{d}{\to} (Q_{ZZ}C + \Psi^*_{ZV})'Q_{ZZ}^{-1}(Q_{ZZ}C + \Psi^*_{ZV})
    \end{align*}
    (\textit{c}) Again we have,
        \begin{align*}
        X'P_Zu &= (Z'X/\sqrt{n})'(Z'Z/n)^{-1}(Z'u/\sqrt{n}) \overset{d}{\to} (Q_{ZZ}C + \Psi^*_{ZV})'Q_{ZZ}^{-1}\Psi^*_{Zu} 
    \end{align*}
    (\textit{d}) Since $ X'X = (ZC/\sqrt{n} + V)'(ZC/\sqrt{n} + V) = C'(Z'Z/n)C + 2C'Z'V/\sqrt{n} + V'V$, it follows that
    \begin{equation*}
        \frac{1}{n}X'X = C'\left(\frac{1}{n^2}Z'Z\right)C + 2C'\frac{1}{n^{3/2}}Z'V + \frac{1}{n} V'V \overset{p}{\to} \mathbb{E}[v_iv_i'] = \tilde{\Sigma}_{V}
    \end{equation*}
    (\textit{e}) Similarly to \textit{d.}, $X'u = C'Z'u\sqrt{n} + V'u$, so
    \begin{equation}
        \frac{1}{n}X'u = C'\frac{1}{n^{3/2}}Z'u + \frac{1}{n}V'u \overset{p}{\to} \mathbb{E}[v_iu_i] = \tilde{\Sigma}_{Vu} \tag*{\QEDA}
    \end{equation}
\end{proof}

\begin{proof}[Proof of Lemma \ref{lem weakhet alphal}]
Given Assumptions \ref{as weakhom xziid}, \ref{as weakhet qzz} and \ref{as weakhet errors}, again with $\bar{\Pi}= [\frac{1}{\sqrt{n}}C\beta \ \  \frac{1}{\sqrt{n}}C ]$ and $B = [\beta \ \ I_{k_x}]$, it follows that
\begin{equation}\label{eq weakhet Z'W}
    \frac{1}{\sqrt{n}}Z'W = \frac{1}{n}Z'ZCB + \frac{1}{\sqrt{n}}Z'\bar{V} \overset{d}{\to} Q_{ZZ}CB + \Psi_{Z\bar{V}}^*
\end{equation}
where $Z'\bar{V}/\sqrt{n} \overset{d}{\to} \Psi_{Z\bar{V}}^*$. Therefore, from (\ref{eq weakhet Z'W}) and $W'W/n\overset{p}{\to}\bar{\Sigma}_{\bar{V}}^*$ we have
\begin{align}
    n\hat{\alpha}_L &=  \min_{||\phi||=1}\frac{\phi'W'P_ZW\phi}{\phi'(\frac{1}{n}W'W)\phi} \notag \\
    & \overset{d}{\to} \min_{||\phi||=1} \frac{\phi'[Q_{ZZ}CB + \Psi_{Z\bar{V}}^*]'Q_{ZZ}^{-1}[Q_{ZZ}CB + \Psi_{Z\bar{V}}^*]\phi}{\phi'\Sigma_{\bar{V}}^*\phi} = \tilde{\alpha}_L^* \tag*{\QEDA}
\end{align}
\end{proof}

\begin{proof}[Proof of Lemma \ref{lem weakhet betalimits}]
Proof follows from Lemmas \ref{lem weakhet staigerstock} and \ref{lem weakhet alphal}. \QEDA
\end{proof}

\begin{proof}[Proof of Lemma \ref{lem weakhet pilimit}]
(\textit{a}) Combining $Z'Z/n\overset{p}{\to}Q_{ZZ}$ from Assumption \ref{as weakhet qzz} and $Z'X/\sqrt{n} \overset{d}{\to} Q_{ZZ}C + \Psi^*_{ZV} $ from Lemma \ref{lem weakhet staigerstock}, we obtain
\begin{equation*}
    \sqrt{n}\hat{\Pi}_{2SLS} \overset{d}{\to} Q_{ZZ}^{-1} (Q_{ZZ}C + \Psi^*_{ZV}).
\end{equation*}
(\textit{b}) The estimator can be written as
\begin{align}\label{eq proofs piliml}
   \hat{\Pi}_{L} &= (Z'Z - Z'\hat{u}_L(\hat{u}_L'\hat{u}_L)^{-1}\hat{u}_L'Z)^{-1}(Z'X - Z'\hat{u}_L(\hat{u}_L'\hat{u}_L)^{-1}\hat{u}_L'X) \notag \\
   &= \left(\frac{1}{n}Z'Z - \frac{1}{n}Z'\hat{u}_L\left(\frac{1}{n}\hat{u}_L'\hat{u}_L\right)^{-1}\frac{1}{n}\hat{u}_L'Z\right)^{-1}\left(\frac{1}{\sqrt{n}}Z'X - \frac{1}{\sqrt{n}}Z'\hat{u}_L\left(\frac{1}{n}\hat{u}_L'\hat{u}_L\right)^{-1}\frac{1}{n}\hat{u}_L'X\right).
\end{align}
We know that $Z'Z/n\overset{p}{\to}Q_{ZZ}$, $Z'X\sqrt{n}\overset{d}{\to}Q_{ZZ}C + \Psi_{ZV}^*$ and $Z'\hat{u}_L/n\overset{p}{\to}0$. Consider the limits
\begin{align*}
    \frac{1}{\sqrt{n}}Z'\hat{u}_L &= \frac{1}{\sqrt{n}}Z'u - \frac{1}{\sqrt{n}}Z'X(\hat{\beta}_L-\beta)\overset{d}{\to} \Psi_{Zu}^* - (Q_{ZZ}C + \Psi_{ZV}^*)\tilde{\beta}^*_L, \\
    \frac{1}{n}\hat{u}_L'\hat{u}_L &= \frac{1}{n}u'u - 2\frac{1}{n}u'X(\hat{\beta}_L-\beta) + (\hat{\beta}_L-\beta)'X'X(\hat{\beta}_L-\beta) \overset{d}{\to} \tilde{\sigma}_{u} - 2\tilde{\Sigma}_{Vu}'\tilde{\beta}^*_L + \tilde{\beta}^{*'}_L\tilde{\Sigma}_{V}\tilde{\beta}^*_L \textup{ and } \\
    \frac{1}{n}X'\hat{u}_L &= \frac{1}{n}X'(u - X(\hat{\beta}_L-\beta)) \overset{d}{\to} \tilde{\Sigma}_{Vu} - \tilde{\Sigma}_{V}\tilde{\beta}^*_L.
\end{align*}
Drawing these results together with (\ref{eq proofs piliml}) yields
\begin{align*}
    \tilde{\Pi}^*_L &= Q_{ZZ}^{-1}\left[Q_{ZZ}C + \Psi_{ZV}^* - \frac{(\Psi_{Zu}^* - (Q_{ZZ}C + \Psi_{ZV}^*)\tilde{\beta}^*_L)(\tilde{\Sigma}_{Vu} - \tilde{\Sigma}_{V}\tilde{\beta}^*_L)'}{\tilde{\sigma}_{u} - 2\tilde{\Sigma}_{Vu}'\tilde{\beta}^*_L + \tilde{\beta}^{*'}_L\tilde{\Sigma}_{V}\tilde{\beta}^*_L}\right] \\
    &= \tilde{\Pi}^*_{2SLS} - Q_{ZZ}^{-1}\left[\frac{(\Psi_{Zu}^* - (Q_{ZZ}C + \Psi_{ZV}^*)\tilde{\beta}^*_L)(\tilde{\Sigma}_{Vu} - \tilde{\beta}^*_L\tilde{\Sigma}_{V})'}{ \tilde{\sigma}_{u} - 2\tilde{\Sigma}_{Vu}'\tilde{\beta}^*_L + \tilde{\beta}^{*'}_L\tilde{\Sigma}_{V}\tilde{\beta}^*_L}\right]. \tag*{\QEDA}
\end{align*}
\end{proof}

\begin{proof}[Proof of Theorem \ref{thm robscore}]
To derive the limiting distribution of the robust score statistic, let
\begin{equation*}
    S_r(\hat{\beta)} = (\sqrt{n}\hat{\theta})'[n\mathbb{V}_r(\hat{\theta})]^{-1}(\sqrt{n}\hat{\theta}).
\end{equation*}
where
\begin{equation*}
    \hat{\theta} = (Z_2' M_{\hat{X}}Z_2)^{-1}Z_2 'M_{\hat{X}}\hat{u},
\end{equation*}
First,
\begin{equation*}
    \frac{1}{\sqrt{n}}Z_2 'M_{\hat{X}}\hat{u} = \frac{1}{\sqrt{n}}Z_2 '\hat{u} = \frac{1}{\sqrt{n}}Z_2'(u -X(\hat{\beta} - \beta)) = \frac{1}{\sqrt{n}}Z_2'u - \frac{1}{\sqrt{n}}Z_2'X(\hat{\beta} - \beta)).
\end{equation*}
From Assumptions \ref{as weakhom xziid}, \ref{as weakhet qzz} and \ref{as weakhet errors} and Lemma \ref{lem weakhet betalimits}, we have
\begin{equation}\label{eq weakhet z2mxu}
    \frac{1}{\sqrt{n}}Z_2'u - \frac{1}{\sqrt{n}}Z_2'X(\hat{\beta} - \beta))\overset{d}{\to} \Psi^*_{2,u} - (Q_2'C + \Psi^*_{2,V})\tilde{\beta}^{*}.
\end{equation}
Further, from Lemma \ref{lem weakhet pilimit} it follows that
\begin{align}\label{eq weakhet z2mxz2}
    \frac{1}{n}Z_2' M_{\hat{X}}Z_2 &= \frac{1}{n}Z_2'Z_2 - \frac{1}{n}Z_2'\hat{X} (\hat{X}'\hat{X})^{-1}\hat{X}'Z_2 \notag \\
    &= \frac{1}{n}Z_2'Z_2 - \frac{1}{n}Z_2'Z\hat{\Pi} (\hat{\Pi}'Z'Z\hat{\Pi})^{-1}\hat{\Pi}'Z'Z_2 \notag \\
    &= \frac{1}{n}Z_2'Z_2 - \frac{1}{n}Z_2'Z\sqrt{n}\hat{\Pi}\left(\sqrt{n}\hat{\Pi}'\frac{1}{n}Z'Z\sqrt{n}\hat{\Pi}\right)^{-1}\sqrt{n}\hat{\Pi}'\frac{1}{n}Z'Z_2 \notag \\
    &\overset{d}{\to} Q_{22} - Q_2'\tilde{\Pi}^*(\tilde{\Pi}^{*'}Q_{ZZ}\tilde{\Pi}^*)^{-1}\tilde{\Pi}^{*'}Q_2.
\end{align}
Combining the results in (\ref{eq weakhet z2mxu}) and (\ref{eq weakhet z2mxz2}), we obtain
\begin{equation}\label{eq weakhet gammalim}
    \sqrt{n}\hat{\theta} \overset{d}{\to} [Q_{22} - Q_2'\tilde{\Pi}^*(\tilde{\Pi}^{*'}Q_{ZZ}\tilde{\Pi}^*)^{-1}\tilde{\Pi}^{*'}Q_2]^{-1}[\Psi^*_{2,u} - (Q_2'C + \Psi^*_{2,V})\tilde{\beta}^{*}].
\end{equation}
For the limiting distribution of the robust score test, we also need the limiting distribution of the robust variance estimator
\begin{equation*}
    \mathbb{V}_r(\hat{\theta}) = \left(\frac{1}{n}Z_2' M_{\hat{X}}Z_2\right)^{-1}\left(\frac{1}{n}\sum_{i=1}^n \hat{u}_i^2 \tilde{z}_{2,i}\tilde{z}_{2,i}'\right)^{-1} \left(\frac{1}{n}Z_2' M_{\hat{X}}Z_2\right).
\end{equation*}
By Assumption \ref{as weakhet ztilde}, it follows that
\begin{equation*}
    \sum_{i=1}^n \hat{u}_i^2 \tilde{z}_{2,i}\tilde{z}_{2,i}' \overset{d}{\to} \Omega_{\tilde{Z}_{2},u} - 2\Omega_{\tilde{Z}_{2},uV\tilde{\beta}} + \Omega_{\tilde{Z}_{2},\tilde{\beta} V\tilde{\beta}} = \Omega_{\tilde{Z}_2\hat{u}}.
\end{equation*}
so the limit of the variance estimator is
\begin{align}\label{eq weakhet gammavarlim}
\begin{split}
    n\mathbb{V}_r(\hat{\theta}) \overset{p}{\to} [Q_{22} - Q_2'\tilde{\Pi}^*(\tilde{\Pi}^{*'}Q_{ZZ}\tilde{\Pi}^*)^{-1}\tilde{\Pi}^{*'}Q_2]^{-1}\Omega_{\tilde{Z}_2\hat{u}}
   \cdot \, [Q_{22} - Q_2'\tilde{\Pi}^*(\tilde{\Pi}^{*'}Q_{ZZ}\tilde{\Pi}^*)^{-1}\tilde{\Pi}^{*'}Q_2]^{-1}.
   \end{split}
\end{align}
From (\ref{eq weakhet gammalim}) and (\ref{eq weakhet gammavarlim}), we get
\begin{equation}\label{eq weakhet scorelimit}
    S_r(\hat{\beta}) \overset{d}{\to} [\Psi^*_{2,u} - (Q_2'C + \Psi^*_{2,V})\tilde{\beta}^{*}]' \Omega_{\tilde{Z}_2\hat{u}}^{-1} [\Psi^*_{2,u} - (Q_2'C + \Psi^*_{2,V})\tilde{\beta}^{*}]
\end{equation}
concluding the proof. \QEDA
\end{proof}

\section{Vectorised results}\label{app vectorised}

Here we present the results above but in vectorised form. The advantage of this form is that results are expressed in terms of multivariate distributions instead of matrix distributions, but at the expense of more cumbersome notation. For simplicity, assume the instruments are normalised such that $Z'Z/n = I_{k_z}$ and $Q_{ZZ} = I_{k_z}$ as in \textcite{montiel2013} and \textcite{lewis2022}. Define the following: $\mcR_{K,L} = I_K \otimes \textup{vec}(I_L)$ is a $KL^2 \times K$ selection matrix, $\Gamma_{ZX} = \textup{vec}(Z'X)/\sqrt{n}$, $\Gamma_{Zy} = Z'y/\sqrt{n}$, $\Gamma_{ZW} = \textup{vec}(Z'W)/\sqrt{n}$, $\Gamma_{WW} = \textup{vec}((W'W)^{1/2})/\sqrt{n}$, $\Gamma_{XX} = \textup{vec}((X'X)^{1/2})/\sqrt{n}$ and $\Gamma_{Xy} = X'y/\sqrt{n}$. It can be shown that
\begin{align*}
    \hat{\beta}_{2SLS} = \left[\mcR_{k_x,k_z}'\left[\left(\Gamma_{ZX}\Gamma_{ZX}'\right)\otimes I_{k_z}\right]\mcR_{k_x,k_z}\right]^{-1} \mcR_{k_x,k_z}'\textup{vec}\left(\Gamma_{Zy}\Gamma_{ZX}'\right)
\end{align*}
and
\begin{align*}
    \hat{\beta}_{L} = \left[\mcR_{k_x,k_z}'\left[\left(\Gamma_{ZX}\Gamma_{ZX}'\right)\otimes I_{k_z}\right]\mcR_{k_x,k_z} - \hat{\alpha}_L\mcR_{k_x,k_x}'\left(\Gamma_{XX}\Gamma_{XX}')\otimes I_{k_x}\right)\right]^{-1}\, \cdot \\
    \left[\mcR_{k_x,k_z}'\textup{vec}\left(\Gamma_{Zy}\Gamma_{ZX}'\right) - \hat{\alpha}_L \mcR_{k_x,k_x}'\textup{vec}\left(\Gamma_{Xy} \otimes I_{k_x}\right) \right]
\end{align*}
where $\hat{\alpha}_L$ is the smallest root of the characteristic polynomial
\begin{align*} 
\left|\mcR_{k_x+1,k_z}'\left[\left(\Gamma_{ZW}\Gamma_{ZW}'\right)\otimes I_{k_z}\right]\mcR_{k_x+1,k_z} - \alpha \mcR_{k_x+1,k_x+1}'\left[\left(\Gamma_{WW}\Gamma_{WW}'\right)\otimes I_{k_z}\right]\mcR_{k_x+1,k_x+1} \right| = 0. 
\end{align*}
Denoting the reduced-form error as $\bar{V} = V\beta + u$, then assume
\begin{align*}
   \begin{pmatrix}
       Z'\bar{V}/\sqrt{n} \\
       \textup{vec}(Z'V)/\sqrt{n}
   \end{pmatrix}  \overset{d}{\to} N(0,\mcW),
   \,\,\, 
   \mcW = 
   \begin{bmatrix}
       \mcW_{1} & \mcW_{12} \\
       \mcW_{21} & \mcW_{2}
   \end{bmatrix}
\end{align*}
where $\mcW$ is a $(k_x+1)k_z \times (k_x+1)k_z$ matrix, with $\mcW_1 = \Omega_{Zu} + (\beta' \otimes I_{k_z})\Omega_{ZV}(\beta \otimes I_{k_z}) - (\beta' \otimes I_{k_z})\Omega_{Z,Vu} - \Omega_{Z,uV}(\beta \otimes I_{k_z})$, $  \mcW_{12} = \Omega_{Zu} -  (\beta' \otimes I_{k_z})\Omega_{Z,Vu}$, $\mcW_{21} = \mcW_{12}'$ and $\mathcal{W}_2 = \Omega_{ZV}$. Then,
\begin{align}\label{eq vec reducedlimit}
   \Gamma_{ZW} = 
   \begin{pmatrix}
       \Gamma_{Zy} \\
       \Gamma_{ZX}
   \end{pmatrix}  \overset{d}{\to} 
   \tilde{\Gamma}_{ZW}^* = 
    \begin{pmatrix}
       \tilde{\Gamma}_{Zy}^* \\
       \tilde{\Gamma}_{ZX}^*
   \end{pmatrix} \sim
   N\left(\mathcal{M}, \mcW
   \right).
\end{align}
where $\mathcal{M} = [(\beta'\otimes I_{k_z})' \,\,\, I_{k_x k_z}]' \textup{vec}(C)$. Further, let $\Gamma_{WW}\overset{d}{\to}\tilde{\Gamma}_{WW}^*$ and $\Gamma_{XX}\overset{p}{\to}\tilde{\Gamma}_{XX}^*$. Then, the limiting distributions from Lemma \ref{lem weakhet betalimits} can be expressed equivalently as
\begin{align}\label{eq vec 2slslimit}
    \tilde{\beta}_{2SLS}^* = \left[\mcR_{k_x,k_z}'\left[\left(\tilde{\Gamma}_{ZX}^*\tilde{\Gamma}_{ZX}^{*'}\right)\otimes I_{k_z}\right]\mcR_{k_x,k_z}\right]^{-1} \mcR_{k_x,k_z}'\textup{vec}\left(\tilde{\Gamma}_{Zy}^*\tilde{\Gamma}_{ZX}^{*'}\right)
\end{align}
and
\begin{align}\label{eq vec limllimit}
    \tilde{\beta}_{L}^* = \left[\mcR_{k_x,k_z}'\left[\left(\tilde{\Gamma}_{ZX}^*\tilde{\Gamma}_{ZX}^{*'}\right)\otimes I_{k_z}\right]\mcR_{k_x,k_z} - \tilde{\alpha}_L^* \mcR_{k_x,k_x}'\left(\tilde{\Gamma}_{XX}^*\tilde{\Gamma}_{XX}^{*'}\otimes I_{k_x}\right)\right]^{-1}\, \cdot \notag \\\left[\mcR_{k_x,k_z}'\textup{vec}\left(\tilde{\Gamma}_{Zy}^*\tilde{\Gamma}_{ZX}^{*'}\right) - \tilde{\alpha}_L^* \mcR_{k_x,k_x}'\textup{vec}\left(\tilde{\Gamma}_{Xy}^* \otimes I_{k_x}\right) \right]
\end{align}
where $\hat{\alpha}_L \overset{d}{\to} \tilde{\alpha}_L^*$, for $\tilde{\alpha}_L^*$ the smallest root of the characteristic polynomial
\begin{align*}  
    \left|\mcR_{k_x+1,k_z}'\left[\left(\tilde{\Gamma}_{ZW}^*\tilde{\Gamma}_{ZW}^{*'}\right)\otimes I_{k_z}\right]\mcR_{k_x+1,k_z} - \alpha \mcR_{k_x+1,k_x+1}'\left[\left(\tilde{\Gamma}_{WW}^*\tilde{\Gamma}_{WW}^{*'}\right)\otimes I_{k_z}\right]\mcR_{k_x+1,k_x+1} \right| = 0. 
\end{align*}
which follow from (\ref{eq vec reducedlimit}).

For the first-stage parameter matrix and estimators, let $\pi = \textup{vec}(\Pi)$ and $\hat{\pi} = \textup{vec}(\hat{\Pi})$. For 2SLS, the first-stage estimator is simply $\sqrt{n}\,\textup{vec}(\hat{\Pi}_{2SLS}) = \sqrt{n}\hat{\pi}_{2SLS} = \Gamma_{ZX}$, and for LIML, we have from \textcite{windmeijer2018} that
\begin{equation}
    \sqrt{n}\,\textup{vec}(\hat{\Pi}_{L}) = \sqrt{n}\,\hat{\pi}_{L} = \sqrt{n} \, \left[ \left(\hat{B}_L'\otimes I_{k_z}\right)' \hat{\mcW}^{-1} \left(\hat{B}_L'\otimes I_{k_z}\right) \right]^{-1}  \left(\hat{B}_L'\otimes I_{k_z}\right)' \hat{\mcW}^{-1} \hat{\pi}_{2SLS}
\end{equation}
for $\hat{B}_L = [\hat{\beta}_L \ \ I_{k_x}]$ the LIML-estimated counterpart of $B = [\beta \ \ I_{k_x}]$ from Lemma \ref{lem weakhet alphal}. From the above results, it therefore follows that
\begin{equation}
    \sqrt{n}\hat{\pi}_{2SLS} \overset{d}{\to} \tilde{\pi}_{2SLS}^* = \tilde{\Gamma}_{ZX}^*
\end{equation}
and
\begin{equation}
    \sqrt{n}\,\hat{\pi}_{L}  \overset{d}{\to} \sqrt{n} \, \left[ \left(\tilde{B}^{*'}_L\otimes I_{k_z}\right)' \mcW^{-1} \left(\tilde{B}^{*'}_L\otimes I_{k_z}\right) \right]^{-1}  \left(\tilde{B}^{*'}_L\otimes I_{k_z}\right)' \mcW^{-1} \tilde{\Gamma}_{ZX}^*
\end{equation}
where $\tilde{B}_L^* = [\tilde{\beta}_L^* \ \ I_{k_x}]$. Consider the matrix $\tilde{Z}_2'\tilde{Z}_2 = Z_2'M_{\hat{X}}Z_2$. This can be re-written as
\begin{align}
    \tilde{Z}_2'\tilde{Z}_2 =  Z_2'Z_2 - &Z_2'Z\left[\left(\hat{\pi}'\otimes I_{k_z}\right)\mcR_{k_x,k_z}\right] \left[\mcR_{k_x,k_z}'\left(\hat{\pi}\otimes I_{k_z}\right) Z'Z\left(\hat{\pi}'\otimes I_{k_z}\right)\mcR_{k_x,k_z}\right]^{-1} \notag \\ \cdot \, &\left[\mcR_{k_x,k_z}'\left(\hat{\pi}\otimes I_{k_z}\right)\right]  Z'Z_2
\end{align}
and from this it follows that
\begin{align}
    \frac{1}{n}\tilde{Z}_2'\tilde{Z}_2 = I_{k_z-k_x} - & J_{k_x,k_z-k_x}\left[\left(\sqrt{n}\hat{\pi}'\otimes I_{k_z}\right)\mcR_{k_x,k_z}\right] \left[\mcR_{k_x,k_z}'\left(\sqrt{n}\hat{\pi}\otimes I_{k_z}\right) I_{k_z} \left(\sqrt{n}\hat{\pi}'\otimes I_{k_z}\right)\mcR_{k_x,k_z}\right]^{-1} \notag \\ \cdot \, &\left[\mcR_{k_x,k_z}'\left(\sqrt{n}\hat{\pi}\otimes I_{k_z}\right)\right]  J_{k_x,k_z-k_x}'
\end{align}
where $Z_2'Z/n = J_{k_x,k_z-k_x} = [0_{k_x \times (k_z-k_x)}' \,\,\, I_{k_z-k_x}]$ is the bottom $k_z-k_x$ rows of $I_{k_z}$. Therefore, $\tilde{Z}_2'\tilde{Z}_2/n$ converges in distribution to
\begin{align}
    \frac{1}{n}\tilde{Z}_2'\tilde{Z}_2 \overset{d}{\to} I_{k_z-k_x} - & J_{k_x,k_z-k_x}\left[\left(\tilde{\pi}^{*'}\otimes I_{k_z}\right)\mcR_{k_x,k_z}\right] \left[\mcR_{k_x,k_z}'\left(\tilde{\pi}^*\otimes I_{k_z}\right) \left(\tilde{\pi}^{*'}\otimes I_{k_z}\right)\mcR_{k_x,k_z}\right]^{-1}\notag  \\ 
    \cdot &\left[\mcR_{k_x,k_z}'\left(\tilde{\pi}^*\otimes I_{k_z}\right)\right] J_{k_x,k_z-k_x}'.
\end{align}
Given these vectorisations, limiting representations of the $J$- and $KP$-tests can then be computed. Consider the robust score in (\ref{eq esttest robscore}), then 
\begin{align*}
    \sqrt{n}\hat{\theta} = \left[\mcR_{k_z-k_x,k_z-k_x}'\left[\left(\Gamma_{\tilde{Z}_2\tilde{Z}_2}\Gamma_{\tilde{Z}_2\tilde{Z}_2}'\right)\otimes I_{k_z-k_x}\right]\mcR_{k_z-k_x,k_z-k_x}\right]^{-1} \mcR_{k_z-k_x,k_z-k_x}'\left(I_{k_z-k_x}\otimes \Gamma_{\tilde{Z}_2\hat{u}}\right)
\end{align*}
with variance
\begin{align*}
    n\hat{\mathbb{V}}(\hat{\theta}) = &\left[\mcR_{k_z-k_x,k_z-k_x}'\left[\left(\Gamma_{\tilde{Z}_2\tilde{Z}_2}\Gamma_{\tilde{Z}_2\tilde{Z}_2}'\right)\otimes I_{k_z-k_x}\right]\mcR_{k_z-k_x,k_z-k_x}\right]^{-1} \\
    \cdot &\left[\mcR_{k_z-k_x,k_z-k_x}'\left(\Gamma_{\tilde{Z}_2\hat{u}}^{\dagger}\otimes I_{k_z-k_x}\right)\left(\Gamma_{\tilde{Z}_2\hat{u}}^{\dagger}\otimes I_{k_z-k_x}\right)'\mcR_{k_z-k_x,k_z-k_x}\right] \\
    \cdot &\left[\mcR_{k_z-k_x,k_z-k_x}'\left[\left(\Gamma_{\tilde{Z}_2\tilde{Z}_2}\Gamma_{\tilde{Z}_2\tilde{Z}_2}'\right)\otimes I_{k_z-k_x}\right]\mcR_{k_z-k_x,k_z-k_x}\right]^{-1}
\end{align*}
where $\Gamma_{\tilde{Z}_2\hat{u}} = \tilde{Z}_2'\hat{u}/\sqrt{n}$ and $\Gamma_{\tilde{Z}_2\hat{u}}^{\dagger}$ is such that $\Gamma_{\tilde{Z}_2\hat{u}}^{\dagger} = \textup{vec}\left(\Gamma_{\tilde{Z}_2\hat{u}}^{\dagger\dagger}\right)$ with $\left(\Gamma_{\tilde{Z}_2\hat{u}}^{\dagger\dagger}\right)'\left(\Gamma_{\tilde{Z}_2\hat{u}}^{\dagger\dagger}\right) = \tilde{Z}_2'\hat{u}\hat{u}'\tilde{Z}_2/n$ is the robust variance estimator for the score statistic. Taking limits of $S_r(\hat{\beta}) = \left(\sqrt{n}\hat{\theta}\right)'\, \left[n\hat{\mathbb{V}}(\hat{\theta})\right]^{-1}\, \left(\sqrt{n}\hat{\theta}\right)$ gives the vectorised form of the robust score limiting distribution as
\begin{align*}
    S_r(\hat{\beta}) \overset{d}{\to} &\left[\left(I_{k_z-k_x}\otimes \tilde{\Gamma}^*_{\tilde{Z}_2\hat{u}}\right)'\mcR_{k_z-k_x,k_z-k_x}\right] \\ \cdot &\left[\mcR_{k_z-k_x,k_z-k_x}'\left(\tilde{\Gamma}_{\tilde{Z}_2\hat{u}}^{\dagger*}\otimes I_{k_z-k_x}\right)\left(\tilde{\Gamma}_{\tilde{Z}_2\hat{u}}^{\dagger*}\otimes I_{k_z-k_x}\right)'\mcR_{k_z-k_x,k_z-k_x}\right]^{-1} \\
    \cdot &\left[\mcR_{k_z-k_x,k_z-k_x}'\left(I_{k_z-k_x}\otimes \tilde{\Gamma}^*_{\tilde{Z}_2\hat{u}}\right)\right]
\end{align*}
where $\Gamma_{\tilde{Z}_2\hat{u}} \overset{d}{\to} \tilde{\Gamma}^*_{\tilde{Z}_2\hat{u}}$, $\Gamma_{\tilde{Z}_2\hat{u}}^{\dagger} \overset{d}{\to} \tilde{\Gamma}_{\tilde{Z}_2\hat{u}}^{\dagger*}$ and $\left(\Gamma_{\tilde{Z}_2\hat{u}}^{\dagger\dagger}\right)'\left(\Gamma_{\tilde{Z}_2\hat{u}}^{\dagger\dagger}\right) \overset{d}{\to} \left(\tilde{\Gamma}_{\tilde{Z}_2\hat{u}}^{\dagger\dagger*}\right)'\left(\tilde{\Gamma}_{\tilde{Z}_2\hat{u}}^{\dagger\dagger*}\right) = \Omega_{\tilde{Z}_2\hat{u}}$. It can be shown that 
\begin{align*}
    \tilde{\Gamma}^*_{\tilde{Z}_2\hat{u}} = N(0_{k_z-k_x},\Omega_{Z_2 u}) + J_{k_x,k_z-k_x}(\tilde{\beta}'\otimes I_{k_x}) \textup{vec}(C), 
\end{align*}
where $\Omega_{Z_2u}$ is the lower-right $(k_z-k_x)\times (k_z-k_x)$ block of $\Omega_{Zu}$, and further
\begin{align*}
    \Omega_{\tilde{Z}_2\hat{u}} = &\left[\mcR_{k_z-k_x,k_z-k_x}'\left(\tilde{\Gamma}_{\tilde{Z}_2u}^{\dagger*}\otimes I_{k_z-k_x}\right)\left(\tilde{\Gamma}_{\tilde{Z}_2u}^{\dagger*}\otimes I_{k_z-k_x}\right)'\mcR_{k_z-k_x,k_z-k_x}\right] \\
    &- \left[\mcR_{k_z-k_x,k_z-k_x}'\left(\tilde{\Gamma}_{\tilde{Z}_2u}^{\dagger*}\otimes I_{k_z-k_x}\right)\left(\tilde{\Gamma}_{\tilde{Z}_2V\tilde{\beta}}^{\dagger*}\otimes I_{k_z-k_x}\right)'\mcR_{k_z-k_x,k_z-k_x}\right] \\
    &- \left[\mcR_{k_z-k_x,k_z-k_x}'\left(\tilde{\Gamma}_{\tilde{Z}_2V\tilde{\beta}}^{\dagger*}\otimes I_{k_z-k_x}\right)\left(\tilde{\Gamma}_{\tilde{Z}_2u}^{\dagger*}\otimes I_{k_z-k_x}\right)'\mcR_{k_z-k_x,k_z-k_x}\right] \\
    &+ \left[\mcR_{k_z-k_x,k_z-k_x}'\left(\tilde{\Gamma}_{\tilde{Z}_2V\tilde{\beta}}^{\dagger*}\otimes I_{k_z-k_x}\right)\left(\tilde{\Gamma}_{\tilde{Z}_2V\tilde{\beta}}^{\dagger*}\otimes I_{k_z-k_x}\right)'\mcR_{k_z-k_x,k_z-k_x}\right]
\end{align*}
where $\tilde{\Gamma}_{\tilde{Z}_2u}^{\dagger*}$ and $\tilde{\Gamma}_{\tilde{Z}_2V\tilde{\beta}}^{\dagger*}$ are the limits of terms adapted suitability from the definitions above. Specialising to 2SLS and LIML yields the $J$- and $KP$-test limiting distributions respectively. 

\section{Design 2 Monte Carlo}\label{sec mc design2}

The model is kept the same as in the first design except that
\begin{equation}\label{eq mc errors2}
    g(z_i) = \sqrt{\exp(\sum_{j=1}^{k_z}\alpha_j z_{ji})}.
\end{equation}
Therefore, error variance is dependent on all the instruments rather than just $z_{1,i}$ and the heteroskedasticity is incorporated through an exponential function of the instruments. We assume that $\alpha_j = \alpha$ for all $j\in\{1,...,k_z\}$ and allow for $\alpha\in\{0.05,0.1,0.2\}$. 

\begin{figure}[ht!] 
    \centering
    \subfloat[$k_z=2$, $\rho = 0.2$]{{\includegraphics[width=8cm]{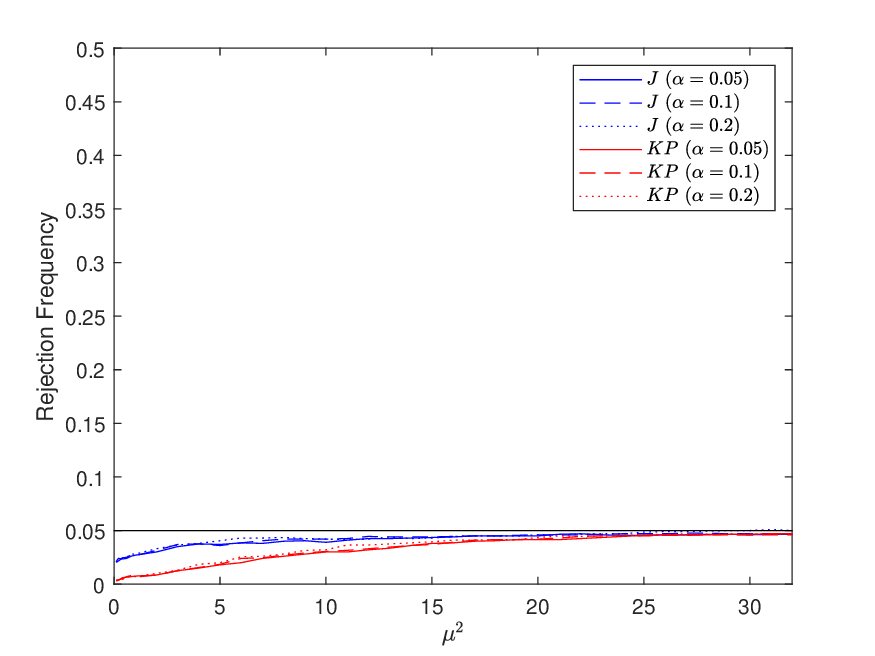} }}
    \subfloat[$k_z=4$, $\rho = 0.2$]{{\includegraphics[width=8cm]{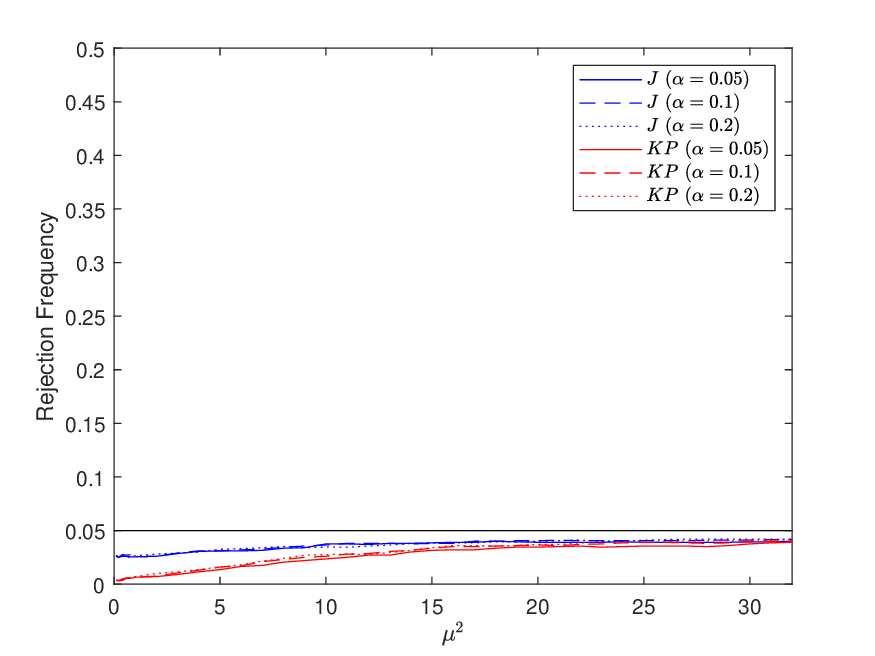} }}
    
    \centering
    \subfloat[$k_z=2$, $\rho = 0.5$]{{\includegraphics[width=8cm]{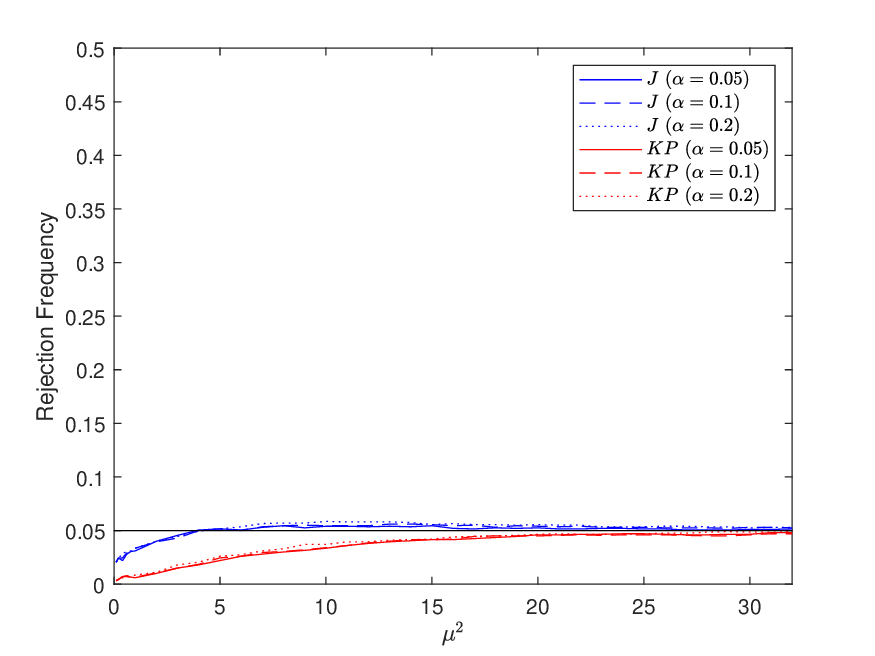} }}
    \subfloat[$k_z=4$, $\rho = 0.5$]{{\includegraphics[width=8cm]{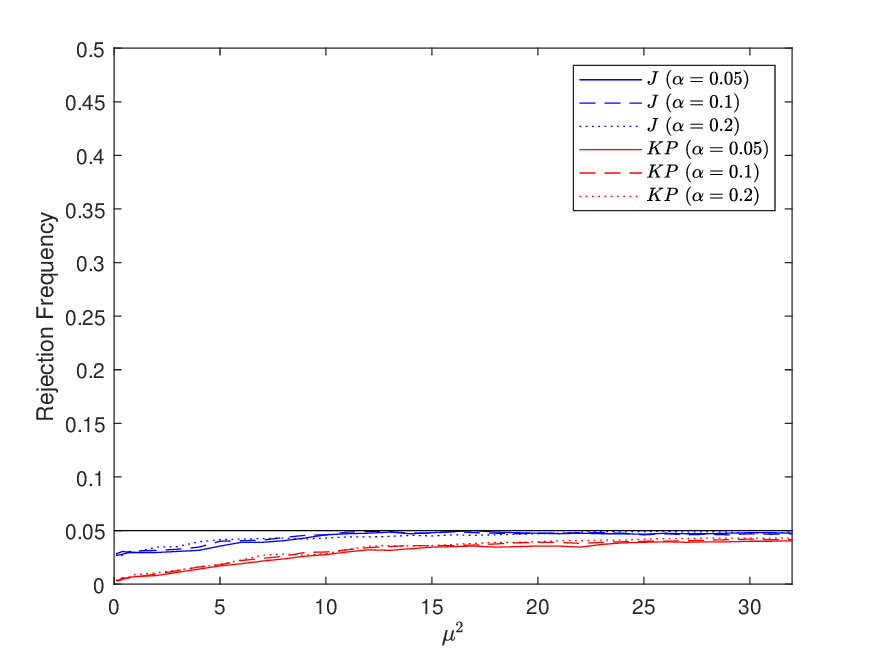} }}
    
    \centering
    \subfloat[$k_z=2$, $\rho = 0.95$]{{\includegraphics[width=8cm]{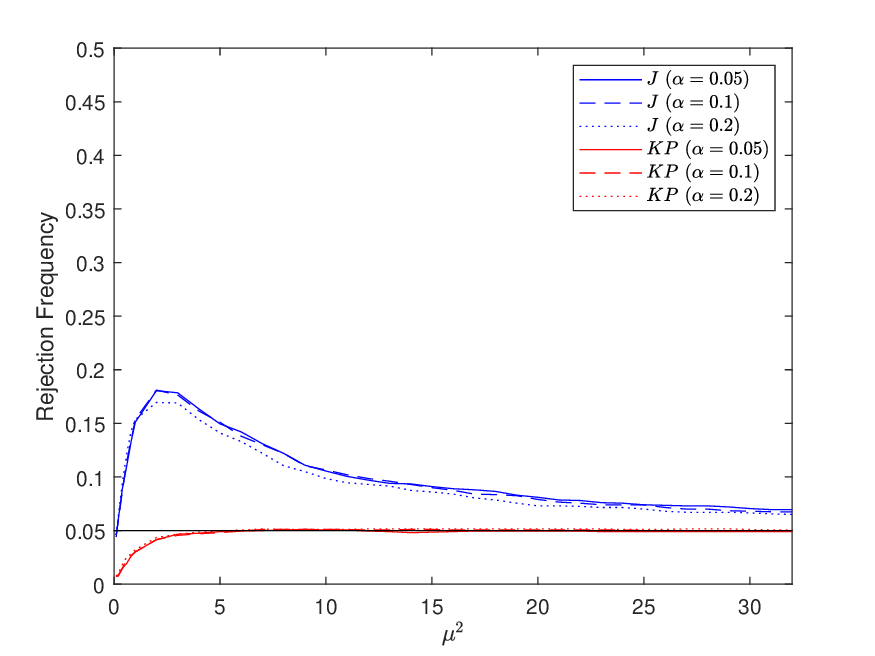} }}
    \subfloat[$k_z=4$, $\rho = 0.95$]{{\includegraphics[width=8cm]{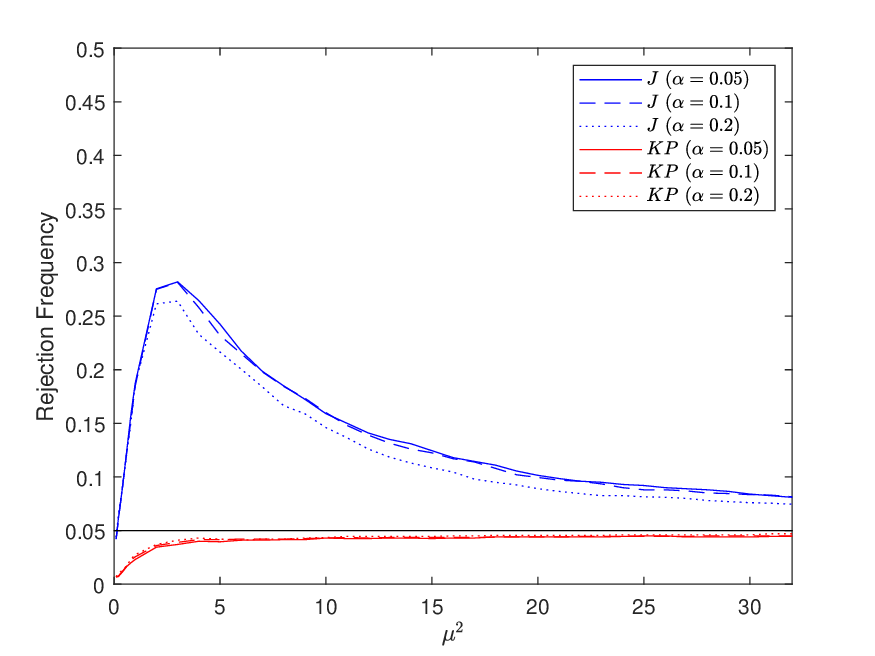} }}
    \caption[Design 2: Rejection frequency at nominal 5\% level across $\mu^2$]{Rejection frequency at nominal 5\% level.}
    \label{graph design2}
\end{figure}

\begin{figure}[h!]
    \centering
    \subfloat[$k_z=2$] {{\includegraphics[width=16cm]{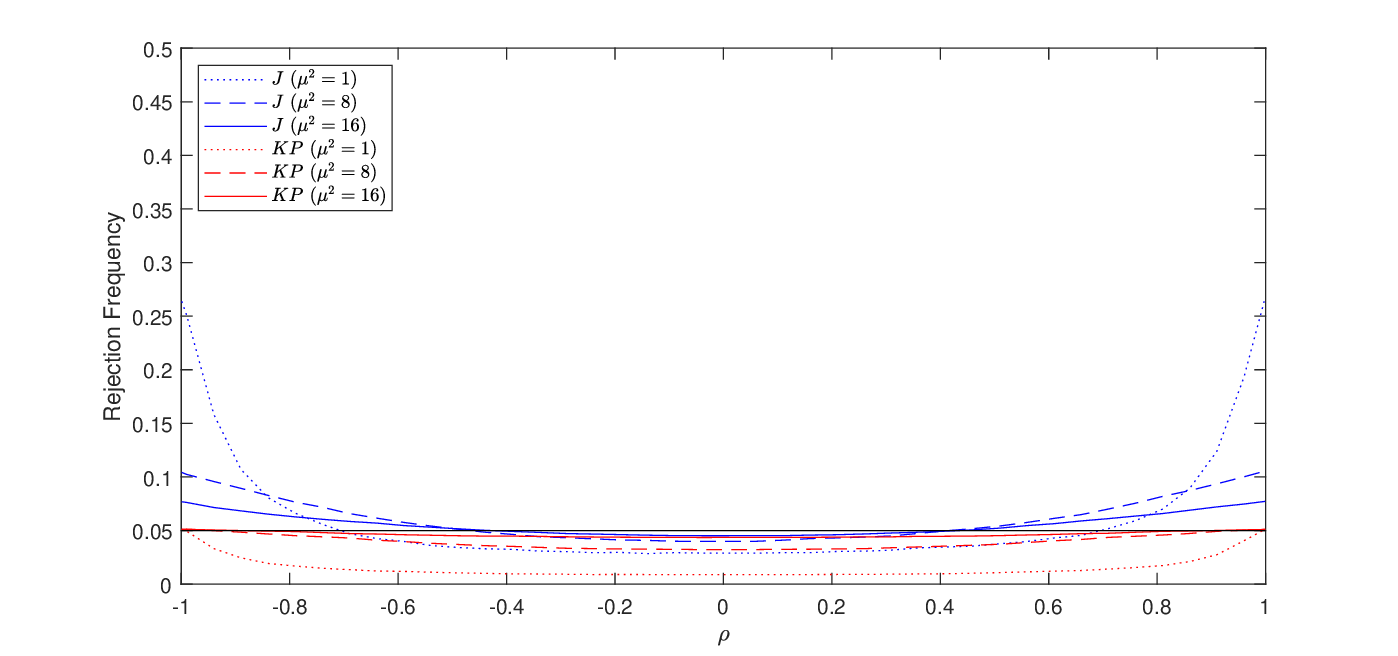}}}

    \centering
    \subfloat[$k_z=4$] {{\includegraphics[width=16cm]{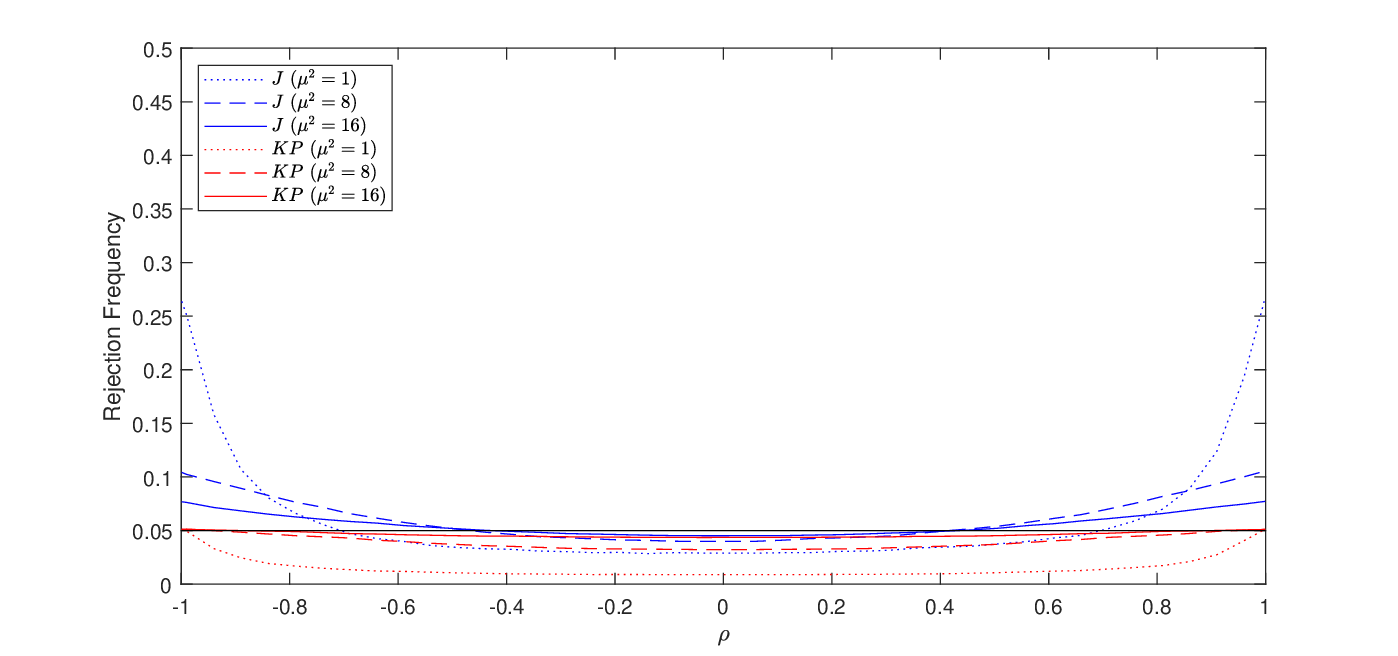}}}
    \caption[Design 2: Rejection frequency at nominal 5\% level across $\rho$]{Rejection frequency at nominal 5\% level across $\rho$.}
    \label{graph design2rho}
\end{figure}

Figure \ref{graph design2} presents similar results to Figure \ref{graph design1}. In the low and medium endogeneity designs, both tests under-reject when identification is extremely weak, with the $J$-test having size closer to the correct nominal value. Both tests however under-reject even for $\mu^2 = 32$. The $J$-test does not change significantly between the $k_z=2$ and $k_z=4$ cases, but the $KP$-test becomes more under-sized; the $J$- and $KP$-test become approximately equally sized at $\mu^2\approx 15$ with one overidentifying restriction, but this increases to about $\mu^2\approx 20$ for $k_z=4$. In the high endogeneity design, we see a similar story to the first heteroskedastic design, with the $J$-test suffering poor performance and the $KP$-test obtaining the correct size for low values of $\mu^2$. In the $k_z=2$ case, $J$ over-rejects with a maximum size of $\approx0.18$ at $\mu^2=2$, whereas $KP$ obtains correct size from $\mu^2\geq 5$. Similarly in the $k_z=6$ case, $J$ over-rejects steeply for low values of $\mu^2$, with size peaking between 0.3-0.35 depending on the strength of the heteroskedasticity. In both designs, $KP$ quickly attains the correct size; the $\mu^2$ required increases with the number of instruments but in both set-ups is still obtained for instrument strengths typically considered weak.

Figure \ref{graph design2rho} is again similar to results from Design 1 presented in Figure \ref{graph design1rho}. For $k_z=2$, the $KP$-test has low size when the instruments are extremely weak, but rejection frequency becomes much closer to the $J$-test when $\mu^2=8$. Other than for the most weak instruments, size distortion is less severe than for $J$. The $J$-test exhibits similar behaviour as seen in Figure \ref{graph design1rho}; size is closer to the nominal 5\% than $KP$ for low values of $|\rho|$ (with the gap increasing in overidentifying restrictions) with the test again over-rejecting substantially as $|\rho|$ becomes large.

\section{Additional estimator and test statistic results (Design 1)}\label{app estimatorresults}

Here we report additional estimator results for Design 1. For 2SLS and LIML, we report the median bias and 90:10 range (defined as the $90^{th}$ percentile minus the $10^{th}$ percentile). These metrics are chosen over the typical bias and variance metrics to allow for a fairer comparison, due the well-known moment problem of LIML. For $J$ and $KP$, we report rejection frequencies at the nominal 10\% and 1\% levels.

\begin{table}[ht!]
\small
\begin{center}
\begin{threeparttable}
\addtolength{\tabcolsep}{0.1pt}
\caption[Estimation and test statistic results for $k_z - k_x = 1$]{Estimation and test statistic results for $k_z - k_x = 1$.} 
\vspace{-0.5em}
\centering 
\begin{tabular*}{\textwidth}{@{\extracolsep{\fill}} c c c c c c c c c c}
\hline
\hline
\multicolumn{10}{c}{Low endogeneity design ($\rho = 0.2$)} \\
\hline
& & \multicolumn{2}{c}{Median Bias} & \multicolumn{2}{c}{90:10 Range} & \multicolumn{2}{c}{Nom. Size 10\%} &  \multicolumn{2}{c}{Nom. Size 1\%}  \T \\\cline{2-3}
\hline
$\alpha$ & $\mu^2$ & 2SLS & LIML & 2SLS & LIML & $J$ & $KP$ & $J$ & $KP$ \\ 
\hline
\hline
& 1 & 0.144 & 0.124 & 2.269 & 4.668 & 0.058 & 0.032 & 0.003 & 0.001 \\
\multirow{2}{*}{0.5} & 4 & 0.070 & 0.039 & 1.520 & 2.448 & 0.074 & 0.057 & 0.005 & 0.003 \\
& 8 & 0.037 & 0.011 & 1.114 & 1.481 & 0.081 & 0.075 & 0.006 & 0.005 \\
& 16 & 0.020 & 0.003 & 0.792 & 0.913 & 0.090 & 0.092 & 0.007 & 0.008 \\
\hline
& 1 & 0.148 & 0.130 & 2.385 & 4.848 & 0.054 & 0.039 & 0.003 & 0.002 \\
\multirow{2}{*}{1} & 4 & 0.081 & 0.051 & 1.667 & 2.862 & 0.070 & 0.068 & 0.004 & 0.006 \\
& 8 & 0.048 & 0.020 & 1.226 & 1.761 & 0.080 & 0.087 & 0.005 & 0.009 \\
& 16 & 0.024 & 0.005 & 0.882 & 1.078 & 0.089 & 0.102 & 0.006 & 0.012 \\
\hline
& 1 & 0.180 & 0.177 & 2.804 & 5.499 & 0.050 & 0.052 & 0.002 & 0.003 \\
\multirow{2}{*}{2} & 4 & 0.134 & 0.110 & 2.297 & 4.408 & 0.059 & 0.075 & 0.002 & 0.008 \\
& 8 & 0.098 & 0.065 & 1.895 & 3.430 & 0.068 & 0.096 & 0.003 & 0.012 \\
& 16 & 0.063 & 0.030 & 1.441 & 2.262 & 0.082 & 0.114 & 0.005 & 0.019 \\
\hline
\hline
\multicolumn{10}{c}{Medium endogeneity design ($\rho = 0.5$)} \\
\hline
& & \multicolumn{2}{c}{Median Bias} & \multicolumn{2}{c}{90:10 Range} & \multicolumn{2}{c}{Nom. Size 10\%} & \multicolumn{2}{c}{Nom. Size 1\%}  \T \\\cline{2-3}
\hline
$\alpha$ & $\mu^2$ & 2SLS & LIML & 2SLS & LIML & $J$ & $KP$ & $J$ & $KP$ \\ 
\hline
\hline
& 1 & 0.355 & 0.296 & 2.056 & 4.425 & 0.066 & 0.035 & 0.005 & 0.001 \\
\multirow{2}{*}{0.5} & 4 & 0.172 & 0.083 & 1.435 & 2.359 & 0.091 & 0.062 & 0.009 & 0.004 \\
& 8 & 0.094 & 0.022 & 1.062 & 1.468 & 0.102 & 0.083 & 0.010 & 0.006 \\
& 16 & 0.046 & 0.003 & 0.774 & 0.913 & 0.103 & 0.095 & 0.011 & 0.009 \\
\hline
& 1 & 0.367 & 0.321 & 2.179 & 4.521 & 0.063 & 0.044 & 0.004 & 0.003 \\
\multirow{2}{*}{1} & 4 & 0.203 & 0.116 & 1.561 & 2.797 & 0.090 & 0.074 & 0.008 & 0.007 \\
& 8 & 0.117 & 0.038 & 1.165 & 1.731 & 0.106 & 0.093 & 0.011 & 0.010 \\
& 16 & 0.058 & 0.007 & 0.854 & 1.074 & 0.109 & 0.103 & 0.012 & 0.012 \\
\hline
& 1 & 0.443 & 0.420 & 2.544 & 4.957 & 0.056 & 0.054 & 0.003 & 0.003 \\
\multirow{2}{*}{2} & 4 & 0.329 & 0.266 & 2.127 & 4.231 & 0.081 & 0.081 & 0.006 & 0.009 \\
& 8 & 0.244 & 0.159 & 1.776 & 3.319 & 0.097 & 0.104 & 0.009 & 0.015 \\
& 16 & 0.157 & 0.070 & 1.368 & 2.205 & 0.112 & 0.119 & 0.015 & 0.020 \\
\hline
\hline
\multicolumn{10}{c}{High endogeneity design ($\rho = 0.95$)} \\
\hline
& & \multicolumn{2}{c}{Median Bias} & \multicolumn{2}{c}{90:10 Range} & \multicolumn{2}{c}{Nom. Size 10\%} &  \multicolumn{2}{c}{Nom. Size 1\%}  \T \\\cline{2-3}
\hline
$\alpha$ & $\mu^2$ & 2SLS & LIML & 2SLS & LIML & $J$ & $KP$ & $J$ & $KP$ \\ 
\hline
\hline
& 1 & 0.647 & 0.444 & 1.395 & 4.010 & 0.238 & 0.082 & 0.071 & 0.005 \\
\multirow{2}{*}{0.5} & 4 & 0.310 & 0.079 & 1.093 & 2.197 & 0.235 & 0.100 & 0.103 & 0.009 \\
& 8 & 0.172 & 0.016 & 0.923 & 1.472 & 0.184 & 0.101 & 0.065 & 0.009 \\
& 16 & 0.085 & 0.002 & 0.718 & 0.921 & 0.146 & 0.102 & 0.036 & 0.009 \\
\hline
& 1 & 0.687 & 0.511 & 1.361 & 3.840 & 0.253 & 0.092 & 0.074 & 0.010 \\
\multirow{2}{*}{1} & 4 & 0.366 & 0.125 & 1.190 & 2.365 & 0.271 & 0.105 & 0.133 & 0.011 \\
& 8 & 0.213 & 0.032 & 0.971 & 1.737 & 0.220 & 0.105 & 0.095 & 0.009 \\
& 16 & 0.108 & 0.005 & 0.773 & 1.102 & 0.167 & 0.104 & 0.053 & 0.009 \\
\hline
& 1 & 0.837 & 0.761 & 1.208 & 3.018 & 0.207 & 0.102 & 0.046 & 0.015 \\
\multirow{2}{*}{2} & 4 & 0.612 & 0.410 & 1.410 & 3.937 & 0.329 & 0.122 & 0.146 & 0.021 \\
& 8 & 0.447 & 0.217 & 1.345 & 3.647 & 0.325 & 0.116 & 0.168 & 0.016 \\
& 16 & 0.286 & 0.083 & 1.101 & 2.045 & 0.280 & 0.113 & 0.142 & 0.012 \\
\hline

\end{tabular*}
\caption*{\footnotesize{Table \ref{table extraresults1}: Additional estimator and test statistic results for 2SLS, LIML, $J$ and $KP$ with one overidentifying restriction. Median bias and 90:10 range (defined as the $90^{th}$ percentile minus the $10^{th}$ percentile) are reported for 2SLS and LIML. Rejection frequencies at the nominal 10\% and 1\% levels are reported for $J$ and $KP$.}}
\label{table extraresults1}
\end{threeparttable}
\end{center}
\end{table}

\begin{table}[ht!]
\small
\begin{center}
\begin{threeparttable}

\caption[Estimation and test statistic results for $k_z - k_x = 3$]{Estimation and test statistic results for $k_z - k_x = 3$.}
\vspace{-0.5em}
\centering 
\begin{tabular*}{\textwidth}{@{\extracolsep{\fill}} c c c c c c c c c c}
\hline\hline
\multicolumn{10}{c}{Low endogeneity design ($\rho = 0.2$)} \\
\hline
& & \multicolumn{2}{c}{Median Bias} & \multicolumn{2}{c}{90:10 Range} & \multicolumn{2}{c}{Nom. Size 10\%} &  \multicolumn{2}{c}{Nom. Size 1\%}  \T \\\cline{2-3}
\hline
$\alpha$ & $\mu^2$ & 2SLS & LIML & 2SLS & LIML & $J$ & $KP$ & $J$ & $KP$ \\ 
\hline
\hline
& 1 & 0.165 & 0.134 & 1.404 & 4.868 & 0.066 & 0.025 & 0.004 & 0.001 \\
\multirow{2}{*}{0.5} & 4 & 0.104 & 0.042 & 1.095 & 2.671 & 0.078 & 0.046 & 0.005 & 0.002 \\
& 8 & 0.068 & 0.012 & 0.878 & 1.540 & 0.084 & 0.066 & 0.005 & 0.004 \\
& 16 & 0.039 & 0.001 & 0.666 & 0.897 & 0.090 & 0.084 & 0.006 & 0.006 \\
\hline
& 1 & 0.166 & 0.136 & 1.499 & 5.011 & 0.063 & 0.037 & 0.003 & 0.001 \\
\multirow{2}{*}{1} & 4 & 0.112 & 0.051 & 1.175 & 2.959 & 0.073 & 0.063 & 0.004 & 0.004 \\
& 8 & 0.074 & 0.014 & 0.938 & 1.761 & 0.081 & 0.083 & 0.004 & 0.007 \\
& 16 & 0.043 & 0.002 & 0.717 & 1.011 & 0.087 & 0.096 & 0.005 & 0.009 \\
\hline
& 1 & 0.180 & 0.158 & 1.804 & 5.489 & 0.055 & 0.047 & 0.002 & 0.002 \\
\multirow{2}{*}{2} & 4 & 0.144 & 0.103 & 1.569 & 4.423 & 0.064 & 0.078 & 0.002 & 0.006 \\
& 8 & 0.116 & 0.058 & 1.349 & 3.364 & 0.073 & 0.101 & 0.003 & 0.012 \\
& 16 & 0.082 & 0.025 & 1.087 & 2.176 & 0.080 & 0.124 & 0.004 & 0.019 \\
\hline
\hline
\multicolumn{10}{c}{Medium endogeneity design ($\rho = 0.5$)} \\
\hline
& & \multicolumn{2}{c}{Median Bias} & \multicolumn{2}{c}{90:10 Range} & \multicolumn{2}{c}{Nom. Size 10\%} &  \multicolumn{2}{c}{Nom. Size 1\%}  \T \\\cline{2-3}
\hline
$\alpha$ & $\mu^2$ & 2SLS & LIML & 2SLS & LIML & $J$ & $KP$ & $J$ & $KP$ \\ 
\hline
\hline
& 1 & 0.408 & 0.318 & 1.282 & 4.515 & 0.076 & 0.028 & 0.005 & 0.001 \\
\multirow{2}{*}{0.5} & 4 & 0.261 & 0.100 & 1.022 & 2.568 & 0.100 & 0.053 & 0.008 & 0.002 \\
& 8 & 0.170 & 0.026 & 0.826 & 1.467 & 0.111 & 0.073 & 0.010 & 0.004 \\
& 16 & 0.097 & 0.003 & 0.637 & 0.880 & 0.112 & 0.089 & 0.010 & 0.006 \\
\hline
& 1 & 0.416 & 0.337 & 1.369 & 4.656 & 0.074 & 0.039 & 0.004 & 0.002 \\
\multirow{2}{*}{1} & 4 & 0.274 & 0.121 & 1.096 & 2.887 & 0.103 & 0.071 & 0.008 & 0.005 \\
& 8 & 0.185 & 0.036 & 0.888 & 1.686 & 0.114 & 0.087 & 0.011 & 0.008 \\
& 16 & 0.109 & 0.004 & 0.685 & 0.988 & 0.115 & 0.098 & 0.012 & 0.009 \\
\hline
& 1 & 0.458 & 0.416 & 1.612 & 5.139 & 0.062 & 0.050 & 0.003 & 0.002 \\
\multirow{2}{*}{2} & 4 & 0.369 & 0.260 & 1.431 & 4.215 & 0.091 & 0.086 & 0.006 & 0.008 \\
& 8 & 0.297 & 0.150 & 1.246 & 3.211 & 0.111 & 0.111 & 0.010 & 0.015 \\
& 16 & 0.211 & 0.061 & 1.023 & 2.086 & 0.126 & 0.129 & 0.014 & 0.020 \\
\hline
\hline
\multicolumn{10}{c}{High endogeneity design ($\rho = 0.95$)} \\
\hline
& & \multicolumn{2}{c}{Median Bias} & \multicolumn{2}{c}{90:10 Range} & \multicolumn{2}{c}{Nom. Size 10\%} &  \multicolumn{2}{c}{Nom. Size 1\%}  \T \\\cline{2-3}
\hline
$\alpha$ & $\mu^2$ & 2SLS & LIML & 2SLS & LIML & $J$ & $KP$ & $J$ & $KP$ \\ 
\hline
\hline
& 1 & 0.768 & 0.451 & 0.696 & 4.142 & 0.329 & 0.072 & 0.092 & 0.004 \\
\multirow{2}{*}{0.5} & 4 & 0.483 & 0.061 & 0.681 & 2.095 & 0.370 & 0.095 & 0.175 & 0.007 \\
& 8 & 0.317 & 0.005 & 0.617 & 1.350 & 0.284 & 0.098 & 0.112 & 0.006 \\
& 16 & 0.185 & 0.000 & 0.535 & 0.832 & 0.207 & 0.099 & 0.056 & 0.006 \\
\hline
& 1 & 0.784 & 0.496 & 0.708 & 3.965 & 0.359 & 0.087 & 0.107 & 0.007 \\
\multirow{2}{*}{1} & 4 & 0.510 & 0.092 & 0.723 & 2.235 & 0.411 & 0.099 & 0.206 & 0.008 \\
& 8 & 0.345 & 0.013 & 0.660 & 1.486 & 0.325 & 0.098 & 0.142 & 0.006 \\
& 16 & 0.205 & 0.001 & 0.572 & 0.923 & 0.238 & 0.097 & 0.075 & 0.005 \\
\hline
& 1 & 0.875 & 0.752 & 0.702 & 2.982 & 0.296 & 0.108 & 0.072 & 0.015 \\
\multirow{2}{*}{2} & 4 & 0.706 & 0.349 & 0.803 & 3.946 & 0.488 & 0.122 & 0.241 & 0.019 \\
& 8 & 0.560 & 0.151 & 0.814 & 2.808 & 0.479 & 0.109 & 0.266 & 0.012 \\
& 16 & 0.397 & 0.043 & 0.763 & 1.744 & 0.404 & 0.096 & 0.207 & 0.007 \\
\hline 
\end{tabular*}
\caption*{\footnotesize{Table \ref{table extraresults2}: Additional estimator and test statistic results for 2SLS, LIML, $J$ and $KP$ with three overidentifying restrictions. Median bias and 90:10 range (defined as the $90^{th}$ percentile minus the $10^{th}$ percentile) are reported for 2SLS and LIML. Rejection frequencies at the nominal 10\% and 1\% levels are reported for $J$ and $KP$.}}
\label{table extraresults2}
\end{threeparttable}
\end{center}
\end{table}

\end{document}